\documentclass[final]{siamltex}

\usepackage{url}
\usepackage{amsmath, amssymb}
\usepackage[mathscr]{eucal}
\usepackage{graphicx}
\usepackage{subfigure}
\usepackage{cite}
\usepackage{bm}
\usepackage{fancyhdr}
\usepackage[usenames,dvipsnames]{color}
\usepackage{enumerate}
\usepackage{xcolor}
\newtheorem{remark}[theorem]{Remark}

\numberwithin{equation}{section}
\numberwithin{equation}{section}
\usepackage{tikz}
\allowdisplaybreaks
\title{Characterization of the equivalent acoustic scattering for a cluster of an extremely large number of small holes} 

\author{Durga Prasad Challa
\thanks{RICAM, Austrian Academy of Sciences,
Altenbergerstrasse 69, A-4040, Linz, Austria.
(Email: durga.challa@oeaw.ac.at)
} 
 \and Andrea Mantile\thanks{ LMR, EA4535
URCA, F\'{e}d\'{e}ration de Recherche ARC Math\'{e}matiques, FR 3399 CNRS. (andrea.mantile@univ-reims.fr)}
\and  Mourad Sini
\thanks{RICAM, Austrian Academy of Sciences,
Altenbergerstrasse 69, A-4040, Linz, Austria.
(Email:mourad.sini@oeaw.ac.at)
}
}
 
\begin{document}
\graphicspath{{Figures-eps/}}
 \maketitle
\begin{abstract}
 We deal with the time-harmonic acoustic waves scattered by a large number of small holes arbitrary distributed in a bounded part of  a homogeneous background. 
   We assume no periodicity in the distribution of these holes.  Using the asymptotic expansions of the scattered field by such a cluster of holes, we show that as their number $M$ grows following the law 
   $M:=M(a):=O(a^{-s}), \; a\rightarrow 0$, the collection of these holes has one of the following behaviors:
   \begin{enumerate}
   \item if $s<1$, then the scattered fields tend to vanish as $a$ tends to zero, i.e. the cluster is a soft one. 
 
   \item if $s=1$, then the cluster behaves as an equivalent medium modeled by a refraction index, supported in a given bounded domain $\Omega$, which 
   is described by certain geometry properties of the holes and their local distribution.  The cluster is a moderate (or intermediate) one.
   
   \item if $s>1$, and satisfies some additional conditions, then the cluster behaves as a totally reflecting extended body, modeled by a bounded and smooth domain $\Omega$, i.e. the incident waves are totally reflected by the surface 
   of this extended body. The cluster is a rigid one.
   \end{enumerate}
   
   Explicit errors estimates between the scattered fields due to the cluster of small holes and the ones 
   due to equivalent media (as the equivalent index of refraction) or the extended body are provided. 
   The first argument is the derivation of the point interaction approximation of the scattered fields generated by clusters of holes with a number 
   $M$ grows following the law $M:=M(a):=O(a^{-s}), \; a\rightarrow 0$ allowing to consider $s>1$. The justification of the equivalent media in the cases $s\leq 1$, is based on a precise analysis of the dominant Foldy-Lax field appearing in the approximation.
   To justify the equivalent media in the case $s>1$, we reduce the study to the estimate of the scattered field generated by potential barriers supported in $\Omega$ and with an amplitude growing at the rate of $a^{1-s}$, 
   as $a<<1$. Let $h:=a^{\frac{s-1}{2}}$ be the semiclasscal parameter, $h<<1$ as $a<<1$. As a key estimate of this analysis, we showed that the $H^{t}(\partial \Omega)$-norm of the total field decays 
   at the rate of $h^{\frac{1}{2}-t}$ which is optimal. To derive such estimates, we propose a new method based on the spectral decomposition of the related Newtonian potential operator.  
   As a by-product, with this method, we can derive the semiclassical resolvent estimates, i.e for $-h^2\Delta -\kappa^2 +V$, at low frequencies, i.e. $\kappa =O(h),\; h<<1$, for positive, compactly supported but not 
   necessarily smooth potentials $V$. 

\end{abstract}

\pagestyle{myheadings}
 \thispagestyle{plain}
 \markboth{D. P. Challa, A. Mantile and M. Sini }{The equivalent acoustic scattering by an extremely large number of holes}

\section{Introduction and statement of the results}\label{Introduction-smallac-sdlp}
\subsection{The acoustic scattered fields generated by a cluster of small holes}

We set $D_m:=\epsilon B_m+z_m$ to be the small bodies characterized by the parameter 
$\epsilon>0$ and the locations $z_m\in \mathbb{R}^3$, $m=1,\dots,M$, where $B_1, B_2,\dots, B_M$ are $M$ open, bounded and simply connected sets in $\mathbb{R}^3$ with Lipschitz boundaries containing the origin.
We assume that the Lipschitz constants of $B_j$, $j=1,..., M$ are uniformly bounded.  
We denote by  $U^{s}$ the acoustic field scattered by the $M$ small and rigid bodies $D_m\subset \mathbb{R}^{3}$ due to 
the incident plane wave $U^{i}(x,\theta):=e^{ikx\cdot\theta}$, 
with the incident direction $\theta \in \mathbb{S}^2$, with $\mathbb{S}^2$ being the unit sphere. Hence the total field $U^{t}:=U^{i}+U^{s}$ satisfies the following exterior Dirichlet problem of the acoustic waves
\begin{equation}
(\Delta + \kappa^{2})U^{t}=0 \mbox{ in }\mathbb{R}^{3}\backslash \left(\mathop{\cup}_{m=1}^M \bar{D}_m\right),\label{acimpoenetrable}
\end{equation}
\begin{equation}
U^{t}|_{\partial D_m}=0,\, 1\leq m \leq M, \label{acgoverningsupport}  
\end{equation}
\begin{equation}
\frac{\partial U^{s}}{\partial |x|}-i\kappa U^{s}=o\left(\frac{1}{|x|}\right), |x|\rightarrow\infty, ~(\text{S.R.C}) \label{radiationc}
\end{equation}
where  $\kappa>0$ is the wave number, $\kappa=2\pi\slash \lambda$, $\lambda$ is the wave length and S.R.C stands for the Sommerfield radiation condition.  The scattering problem 
(\ref{acimpoenetrable}-\ref{radiationc}) is well posed in appropriate spaces, see \cite{C-K:1998, Mclean:2000} for instance, and the scattered field $U^s(x, \theta)$ has the following asymptotic expansion:
\begin{equation}\label{far-field}
 U^s(x, \theta)=\frac{e^{i \kappa |x|}}{|x|}U^{\infty}(\hat{x}, \theta) + O(|x|^{-2}), \quad |x|
\rightarrow \infty,
\end{equation}
with $\hat{x}:=\frac{x}{\vert x\vert}$, where the function
$U^{\infty}(\hat{x}, \theta)$ for $(\hat{x}, \theta)\in \mathbb{S}^{2} \times \mathbb{S}^{2}$  is called the far-field pattern.
We recall that the fundamental solution, $\Phi_\kappa(x,y)$, of the Helmholtz equation in $\mathbb{R}^3$
with the fixed wave number $\kappa$ is given by
\begin{eqnarray}\label{definition-ac-small-fundamentalkappa}
 \Phi_\kappa(x,y)&:=&\frac{e^{i\kappa|x-y|}}{4\pi|x-y|},\quad \text{for all } x,y\in\mathbb{R}^3.
\end{eqnarray}
\begin{definition} 
\label{Def1}
We define  $ a:=\max\limits_{1\leq m\leq M } diam (D_m) ~~\big[=\epsilon \max\limits_{1\leq m\leq M } diam (B_m)\big]$,
$d:=\min\limits_{\substack{m\neq j\\1\leq m,j\leq M }} d_{mj},
$ where $d_{mj}:=dist(D_m, D_j)$. We assume that $
0\,<\,d\,\leq\,d_{\max}$,
and $d_{\max}$ is given. Finally, we assume that we have $\kappa_{\max}$ as the upper bound of the used wave numbers, i.e. $\kappa\in[0,\,\kappa_{\max}]$.
\end{definition}
\bigskip

Let us now focus on the case where we have a large number of obstacles of the form  $M:=M(a)=O(a^{-s})$ with  $s>0$
and the minimum distance $d:=d(a):=O(a^t)$ with  $t>0$. In \cite{C-S:2014}, we have shown using the single layer representation of the solution that if $s\leq 1-2t$, 
then we have the asymptotic expansion the far-field pattern $U^\infty(\hat{x},\theta)$. Using also the single layer representation of the solution, this condition is weakened to (\ref{conditions-1pr}) 
in the following propostion. In particular (\ref{conditions-1pr}) allows to take $s> 1$ and study the scattering by an extremely large number of holes which is the main topic of this work.

\begin{proposition}\label{Prpopostion-Expansion-Far-fields} Assume that $s$ and $t$ satisfy the following conditions
 \begin{eqnarray} \label{conditions-1pr}
   0\leq{t}\leq{1}&\mbox{ and } 0\leq{s}\leq\min\left\{2,3t,{2-t}, {3-3t},\frac{1}{2}(3-t)\right\}
\end{eqnarray} 
then
\begin{eqnarray}\label{x oustdie1 D_m farmain-recent**}
U^\infty(\hat{x},\theta)
&=&\sum_{m=1}^{M}e^{-i\kappa\hat{x}\cdot z_{m}}{Q}_m+O\left({a^{2-s}+a^{3-2s-t}+ a^{4-3s-t}}+{a^{3-s-2t}}+{a^{4-2s-3t}}\right)
\end{eqnarray}
uniformly in $\hat{x}$ and $\theta$ in $\mathbb{S}^2$. The coefficients $Q_m$, $m=1,..., M,$ are the solutions of the following linear algebraic system
\begin{eqnarray}\label{fracqcfracmain}
 Q_m +\sum_{\substack{j=1 \\ j\neq m}}^{M}C_m \Phi_\kappa(z_m,z_j)Q_j&=&-C_mU^{i}(z_m, \theta),~~
\end{eqnarray}
for $ m=1,..., M,$ with $C_m:=\int_{\partial D_m}\sigma_m(s)ds$ and $\sigma_{m}$ is 
the solution of the integral equation of the first kind
\begin{eqnarray}\label{barqcimsurfacefrm1main}
\int_{\partial D_m}\frac{\sigma_{m} (s)}{4\pi|t-s|}ds&=&1,~ t\in \partial D_m.
\end{eqnarray}
The algebraic system \eqref{fracqcfracmain} is invertible under the conditions:
\begin{eqnarray}\label{invertibilityconditionsmainthm}
 \frac{a}{d}\leq c_1 \text{ and } 
 \min_{j\neq m}\cos(\kappa \vert z_j- z_m\vert)\geq 0,
\end{eqnarray}
where $c_1$ depends only on the Lipschitz character of the obstacles $B_j$, $j=1, ..., M$. 
\end{proposition}

\bigskip

\subsection{The distribution of the small holes}

Let $\Omega$ be a bounded domain, say of unit volume, containing the obstacles $D_m, m=1, ..., M$. 
We shall divide $\Omega$ into $[a^{-s}]$ subdomains $\Omega_m,\; m=1, ..., [a^{-s}]$, \footnote{As an example, taking $a:=N^{-\frac{1}{s}}$, with $N$ an integer and $N>>1$, we have $a<<1$ and $[a^{-s}]=N$.} such that each 
$\Omega_m$ contains $D_m$,i.e. $z_m \in \Omega_m$, and some of the other $D_j$'s. It is natural then to assume that the number of obstacles in $\Omega_m$, 
for $m=1, ..., [a^{-s}]$, to be uniformly bounded in terms of $m$. To describe correctly this number of obstacles, we introduce the function 
 $K: \mathbb{R}^3\rightarrow \mathbb{R}$ as a positive continuous and bounded potential. 
 Let each $\Omega_m$, $m\in \mathbb{N}$, be a cube of volume $a^s\frac{[K(z_m)+1]}{K(z_m)+1}$
 and contains $[K(z_m) +1]$ obstacles, see Fig.\ref{distribution-obstacles}. We set $K_{max}:=\sup_{z_m}(K(z_m) +1)$, hence $M=\sum^{[a^{-s}]}_{j=1}[K(z_m) +1]\leq K_{max}[a^{-s}]=O(a^{-s})$.  
Observe that for the cubes $\Omega_m$'s intersecting with $\partial \Omega$, the sets $\Omega \cap \Omega_m$ will have volumes of the order $a^s$ but the exact volume is not easy to estimate as it depends strongly on the shape of $\Omega$ 
(unless if $\Omega$ has a simple shape itself). We do not put obstacles in those cubes touching $\partial \Omega$.

 As $\Omega$ can have an arbitrary shape, the set of the cubes intersecting $\partial \Omega$ is not empty (unless if $\Omega$ has a simple shape as a cube). 
 Later in our analysis, we will need the estimate of the volume of this set. 
Since each $\Omega_j$ has volume of the order $a^s$, and then its maximum radius is of the order $a^{\frac{1}{3}s}$, then the intersecting surfaces with $\partial \Omega$ has an area of the order $a^{\frac{2}{3}s}$.
As the area of $\partial \Omega$ is of the order one, we conlude that the number of such cubes will not exceed the order $a^{-\frac{2}{3}s}$. Hence the volume of this set will not exceed the order 
$a^{-\frac{2}{3}s}a^{s}=a^{\frac{1}{3}s}$, as $a \rightarrow 0$.

\begin{figure}
\centering
\input{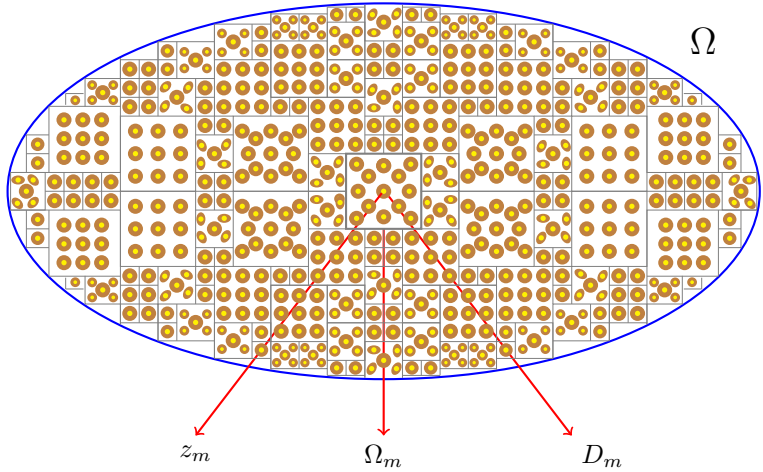}
\caption{A Schematic example on how the obstacles are distributed in $\Omega$.}\label{distribution-obstacles}
\end{figure}


\subsection{The equivalent media}

By applying the estimate (\ref{mazya-fnlinvert-small-ac-3-effect}) of Lemma \ref{Mazyawrkthm-effect} to the system \eqref{fracqcfracmain} and knowing that the capacitances $C_m$ behave as $a$, $a<<1$, 
we have the upper bound
\begin{equation}
 \sum_{m=1}^{M}e^{-i\kappa\hat{x}\cdot z_m}Q_m=O(a^{1-s}).
\end{equation}
 We distinguish the following cases:

\subsubsection{ Case $s<1$} If the number of obstacles is $M:=M(a):=a^{-s}, \; s<1$ and $t$ satisfies (\ref{conditions-1pr})
  , $a\rightarrow 0$, then from (\ref{x oustdie1 D_m farmain-recent**}), we deduce that
\begin{equation}\label{s-smaller-1}
 U^\infty(\hat{x},\theta)\rightarrow 0, \mbox{ as } a \rightarrow 0, \mbox{ uniformly in terms of } \theta \mbox{ and } \hat{x} \mbox{ in } \mathbb{S}^2.
\end{equation}
This means that this collection of obstacles has no effect on the homogeneous medium as $a \rightarrow 0$.

\subsubsection{ Case $s=1$} 

In this case we divide the bounded domain $\Omega$ in the way we explained above. 
In \cite{A-C-K-S-2015}, we proved the following result. Let the small obstacles be distributed in a bounded domain $\Omega$, say of unit volume, with their number $M:=M(a):=O(a^{-1})$ and their minimum distance 
$d:=d(a):=a^{t}$,\; $\frac{1}{3}\leq t \leq \frac{5}{12}$, as $a\rightarrow 0$, as described above. Then
\begin{enumerate}
 \item if the obstacles are distributed arbitrary in $\Omega$, i.e. with different capacitances, then there exists a potential $\bold{C}_0 \in \cap_{p\geq 1}L^p(\mathbb{R}^{3})$ with support in $\Omega$ such that
 \begin{equation}\label{B}
  \lim_{a\rightarrow 0}U^\infty(\hat{x},\theta)= U_{0}^\infty(\hat{x},\theta) \mbox{ uniformly in terms of } \theta \mbox{ and } \hat{x} \mbox{ in } \mathbb{S}^2
 \end{equation}
where $U_{0}^\infty(\hat{x},\theta)$ is the farfield corresponding to the scattering problem

\begin{equation}
(\Delta + \kappa^{2}-(K+1)\bold{C}_0)U_{0}^{t}=0 \mbox{ in }\mathbb{R}^{3},\label{B-1}
\end{equation}
\begin{equation}
U_{0}^{t}=U_{0}^s +e^{i\kappa x\cdot \theta}  
\end{equation}
\begin{equation}
\frac{\partial U_{0}^{s}}{\partial |x|}-i\kappa U_{0}^{s}=o\left(\frac{1}{|x|}\right), |x|\rightarrow\infty. \label{radiationc-B-1}
\end{equation}
 
\item if in addition $K\mid_{\Omega}$ is in $C^{0, \gamma}(\Omega)$, $\gamma \in (0, 1]$ and the obstacles have the same capacitances, then
\begin{equation}\label{C}
  U^\infty(\hat{x},\theta)= U_0^\infty(\hat{x},\theta) +O(a^{\min\{\gamma, \frac{1}{3}-\frac{4}{5}t\}}) \mbox{ uniformly in terms of } \theta \mbox{ and } \hat{x} \mbox{ in } \mathbb{S}^2
 \end{equation}
 where $C_0=C$ in $\Omega$ and $C_0=0$ in $\mathbb{R}^{3} \setminus{\overline \Omega}$.
\end{enumerate}

This result shows the 'equivalent' behaviour between a cluster of, appropriately dense, small holes
and an extended penetrable obstacle modeled by an additive potential. Such an observation goes back at least to the works by Cioranescu and Murat \cite{C-M:1979, C-M:1997} and also
the reference therein. Their analysis, made for the Poisson problem, is based on homogenization via energy methods and, in particular, they assume that the obstacles are distributed periodically. 
In \cite{A-C-K-S-2015}, we have confirmed this result without the periodicity assumption and provided the error of the approximation. In this work, compared to \cite{A-C-K-S-2015}, the error of the approximation $a^{\min\{\gamma, \frac{1}{3}-\frac{4}{5}t\}}$ and the interval where $t$, i.e. $\frac{1}{3}\leq t \leq \frac{5}{12}$, 
are improved to $O\left({ a^{\min\{\frac{1}{3}, 2-3t}}\right),$ and $\frac{1}{3}\leq{t}<\frac{2}{3}$ respectively, at least for the case $K=0$. For instance, for $t=\frac{1}{3}$, the former error is 
of the order $a^{\frac{1}{12}}$ while the latter is of the order $a^{\frac{1}{3}}$.

\subsubsection{Case $s>1$} Here we assume, for simplicity, that the density function $K$ is a trivial function $K=0$ and the holes have the same capacitance, i.e. they have the same shape for instance.
The main contribution of this work is to prove the following result.
 \begin{theorem}\label{C-Th} Let $\Omega$ be bounded domain of class $C^{1,1}$. \footnote{The $C^{1,1}$ regularity is needed in Section \ref{Step-2:Semiclassical-estimates} where semi-classical estimates are derived. 
 This regularity is not needed for the case $s\leq 1$, as it is shown also in \cite{A-C-K-S-2015}.} We divide it in terms of the $\Omega_m$'s as described above.
 In addition, we assume that $\kappa^2$ is not a Dirichlet-Laplacian eigenvalue for $\Omega$. Under the conditions
  \begin{eqnarray}\label{conditions-holes} 
     \frac{s}{3}\leq{t}<{1}&\mbox{ and } 0\leq{s}<\min\left\{(\frac{33}{29})_-,{2-t}, {3-3t},\frac{1}{2}(3-t),\frac{4-t}{3},\frac{4-3t}{2}\right\}
  \end{eqnarray}
 we have the following expansion: 
  \begin{equation}\label{C1}
  U^\infty(\hat{x},\theta)- U_{D}^\infty(\hat{x},\theta)=O\left({ a^{\frac{s-1}{4}}+a^{\frac{(33-29s)_-}{12}}}+{a^{4-3s-t}}+{a^{4-2s-3t}}\right)~~ \mbox{ uniformly in terms of } \theta \mbox{ and } \hat{x} \mbox{ in } \mathbb{S}^2
 \end{equation}
where $U_{D}^\infty(\hat{x},\theta)$ is the farfield corresponding to the Dirichlet scattering problem

\begin{equation}
(\Delta + \kappa^{2})U_{D}^{t}=0 \mbox{ in }\mathbb{R}^{3}\setminus{\overline{\Omega}},\label{C-2}
\end{equation}
\begin{equation}
U_{D}^{t}:=U_{D}^s +e^{i\kappa x\cdot \theta}=0 \mbox{ on } \partial \Omega,  
\end{equation}
\begin{equation}
\frac{\partial U_{D}^{s}}{\partial |x|}-i\kappa U_{D}^{s}=o\left(\frac{1}{|x|}\right), |x|\rightarrow\infty. \label{radiationc-C-1}
\end{equation}
 For $\beta$, a real and positive number, we used the notation $\beta_-$ to describe the property $\alpha < \beta_-$if $\alpha  \leq \beta \; - r $ for some small $r>0$. 
\end{theorem}
 
\bigskip

To have an idea on the order of convergence, we take $t=\frac{s}{3}$, then if $s<(\frac{33}{29})_- $ the conditions in (\ref{conditions-holes}) are satisfied.
 Now, choosing $s:=1.1$, then the error is approximately of the order $a^{\frac{1}{10}}$, as $a<<1$.
\bigskip

To our best knowledge, this kind of result have never been published before even if in few references, as \cite{C-M:1979, C-M:1997} and also
the cited reference therein, it is claimed that with such a dense cluster of holes the scattered fields should behave as the one of the exterior problem. 

\subsection{A brief description of the arguments}
Let us now give our arguments why the claimed results make sense. We recall that $C_j=\bar{C}_j a$ where $\bar{C}_j$ is the capacitance of the reference obstacle $B_j$ and 
we rewrite (\ref{fracqcfracmain}) as:
\begin{eqnarray}\label{fracqcfracmain-s>1-}
 \frac{Q_m}{C_m} +a^{1-s}\sum_{\substack{j=1 \\ j\neq m}}^{M}\bar{C}_j a^{s}\; \Phi_\kappa(z_m,z_j)\frac{Q_j}{C_j}&=&-U^{i}(z_m, \theta),~~
\end{eqnarray}
for $ m=1,..., M$.
Similarly, we rewrite the representation (\ref{x oustdie1 D_m farmain-recent**}) as
\begin{eqnarray}\label{x oustdie1 D_m farmain-recent***}
U^\infty(\hat{x},\theta)
\hspace{-.05cm}=\hspace{-.1cm}\; a^{1-s}\;\sum_{m=1}^{M}e^{-i\kappa\hat{x}\cdot z_m}\bar{C_m} a^s \frac{Q_m}{C_m}\hspace{-.03cm}+\hspace{-.03cm}o(1), \; a \rightarrow 0. 
 \end{eqnarray}
We assume, for simplicity, that the reference  obstacles are all the same, precisely 
they have the same capacitance $\bar{C}_m=\bar{C}_j,\; j\neq m$ and we set $\bar{C}$ that constant. Let us introduce the Lippmann-Schwinger equation
\begin{equation}\label{LS-equation}
 U_a +a^{1-s}\int_{\Omega}\Phi_{\kappa}(\cdot, z)\bar{C}(K(z)+1)U_a(z)\; dz\; =u^{i}(\cdot, \theta).
\end{equation}
The equation (\ref{LS-equation}) is the integral equation modeling the unique solution of the acoustic problem
\begin{equation}\label{Acoustic-model}
 -\Delta U_a -\kappa^2 U_a +a^{1-s}\bar{C}(K(z)+1)U_a\; =0
\end{equation}
where $U_a(\cdot, \theta):=U^s_a(\cdot, \theta)+u^{i}(\cdot, \theta)$ and $U^s_a(\cdot, \theta)$ satisfies the Sommerfeld radiation condition. 
The far-field corresponding to the solution of (\ref{LS-equation}) (or (\ref{Acoustic-model})) has the form
\begin{equation}\label{far-field-acoustic}
 U^{\infty}_a(\hat{x}, \theta):=-a^{1-s}\int_{\Omega}e^{-i\kappa \hat{x}\cdot z}\bar{C}(K(z)+1)U_a(z)dz.
\end{equation}
Based on (\ref{fracqcfracmain-s>1-}-\ref{x oustdie1 D_m farmain-recent***}-\ref{LS-equation}), noticing that $M=O(a^{-s})$ and $\vert \Omega_m\vert=O(a^s),\; s>1$, $m=1,..., M$, we derive the following
error estimate:

\begin{equation}\label{acoustic-difference-farfield}
  U^\infty(\hat{x},\theta)-U^\infty_{a}(\hat{x},\theta) = o(1),\; a \rightarrow 0
\end{equation}
with an explicit expression of the error term $o(1)$.

Our final step is to estimate the term $U^\infty_{a}(\hat{x},\theta),$ as $a\rightarrow 0$. For $s<1$, it is clear from (\ref{far-field-acoustic}) that $U^\infty_{a}$ tends to zero, as $a\rightarrow 0$ uniformly in terms of $\hat{x}$ and $\theta$. For $s=1$, we see that $U^\infty_{a}$
is the farfield corresponding to the scattering problem by the potential $(K+1){\bf{C}}_0$  described in (\ref{B-1})-(\ref{radiationc-B-1}). The most delicate case is when $s>1$. In this case, we need to study the scattering problem (\ref{Acoustic-model}) modeled by the Schr\"{o}dinger equation with the potential $ a^{1-s}\bar{C}(K+1) 1_{\Omega}.$
We set $h:=a^{\frac{1-s}{2}}$ and $V_0:=\bar{C}(K+1)1_{\Omega}$, hence (\ref{Acoustic-model}) becomes
\begin{equation}\label{semiclassical-Schroedinger}
 -\Delta u^t -\kappa^2 u^t +h^{-2}V_0 u^t=0\; \text{ in } \mathbb{R}^3
\end{equation}
with $u^t:=u^t(\cdot, \theta):=u^s(\cdot, \theta) + u^i(\cdot, \theta)$ where $u^s$ ($:=U^s_a(\cdot, \theta)$) satisfies the Sommerfeld radiation conditions.
\bigskip

When $h \rightarrow 0$, the tunneling effect through the potential barrier $h^{-2}V_01_{\Omega}$ becomes negligeable, hence in this limit our model corresponds to a hard obstacle supported in $\Omega$ and 
does not allow any source coming from outside to penetrate inside $\Omega$, as $h \rightarrow 0$. 
We find it quite interesting to link the scattering by a collection of impenetrable obstacles $D_m$, $m=1,...M$, to the scattering by potential barriers.

To derive our result with the error estimates in terms of $h$ and hence justify our claim, we prove the following trace estimate:
\begin{equation}\label{key-estimate}
 \Vert u^t(\cdot, \theta) \Vert_{H^t(\partial \Omega)} =O(h^{\frac{1}{2}-t}),\; t\in [0, \frac{1}{2}], \mbox{ as } h<<1.
\end{equation}

Based on this estimate, we conclude as follows. Taking $t=0$ in (\ref{key-estimate}), we have the estimate $\left\Vert \left.  u^t(\cdot, \theta)  \right\vert _{\partial
\Omega}\right\Vert _{L^2\left(  \partial\Omega\right)}=O(h^{\frac{1}{2}})=O(a^{\frac{s-1}{4}})$, as $a<<1$. As $u^s(\cdot, \theta)$ satisfies $(\Delta +k^2)u^s(\cdot, \theta)=0$ 
in $\mathbb{R}^3\setminus \bar{\Omega}$ with the radiation conditions and the boundary estimate $\left\Vert \left.  u^s(\cdot, \theta) + e^{i\kappa x\cdot \theta} \right\vert _{\partial
\Omega}\right\Vert _{L^2\left(  \partial\Omega\right)}=O(a^{\frac{s-1}{4}})$ as $a<<1$, then the well posedness of the forward scattering problem in the exterior domain $\mathbb{R}^3\setminus \bar{\Omega}$ implies that
the corresponding far-fields satisfy the estimate, recalling that $u^s(\hat{x}, \theta)=U^s_a(\hat{x}, \theta)$,
\begin{equation}\label{acoustic-difference-farfield-1}
\vert U^\infty_a(\hat{x}, \theta)-U^\infty_D(\hat{x}, \theta)\vert = O(a^{\frac{s-1}{4}}) \mbox{ as } a<<1.
\end{equation}
The claim follows by combining (\ref{acoustic-difference-farfield}) and (\ref{acoustic-difference-farfield-1}).
\bigskip

Let us now describe the main idea to derive the estimate (\ref{key-estimate}). For this, we write (\ref{LS-equation}) in the form
\begin{equation}\label{Rk-LS}
(h^2+R(\kappa))u^t=h^2 u^i, \mbox{ in } \Omega
\end{equation}
where $R(\kappa)$, defined through $R(\kappa)u^t(x):=\int_{\Omega}V_0(y)\Phi_\kappa(x, y)u^t(y)dy$, is the Newtonian potential operator restricted to $\Omega$. We rewrite (\ref{Rk-LS}) as 
\begin{equation}\label{R0-LS}
(h^2+R(0))u^t=h^2 u^i - (R(\kappa)-R(0))u^t \mbox{ in } \Omega.
\end{equation}
For simplicity of exposition, we set $V_0\equiv 1$. The key observation is that $(\lambda, e)$ is an eigenvalue and eigenfunction, respectively, of
$R(0)$ if and only if $(\lambda^{-1}, e)$ is  an eigenvalue and eigenfunction of the operator defined, on $\Omega$, by $(-\Delta)$ with the boundary condition $\partial_\nu e-S^{-1}(-\frac{1}{2}I+K)e=0$ on $\partial \Omega$ where $S$ and $K$ are the single and double layer potential respectively. Using variational formulations with boundary integral equations, we show that this second operator is self-adjoint with a compact inverse and allows the following characterization of $H^1(\Omega)$, i.e. $H^1(\Omega)=\{u\in L^2(\Omega),\; \mbox{ such that } \sum_n \lambda_n^{-1}(u, e_n)^2 <\infty\}$, with the spectral family $(\lambda_n, e_n), n\in \mathbb{N}$ (of $R(0)$). Using this characterization and the spectral decomposition (Galerkin method) in (\ref{R0-LS}), we show that $\Vert u^t\Vert_{H^1(\Omega)}=O(1),\; h<<1$. With this estimate and an integration by parts in (\ref{semiclassical-Schroedinger}), in $\Omega$, we deduce that $\Vert u^t\Vert_{L^2(\Omega)}=O(h),\; h<<1$. We end up with interpolation and trace estimates.
\bigskip

We finish this introduction by making a link between the estimates we derived and the semiclassical resolvent estimates. Recall that the solution $u^t=u^i+u^s$ of the equation (\ref{Rk-LS}) is nothing but the solution of the equation 
$(-\Delta -\kappa^2 +h^{-2}V_01_{\Omega})u^t=0$, and then the scattered field $u^s$ is the solution of the equation 
$(-\Delta -\kappa^2 +h^{-2}V_01_{\Omega})u^s=-h^{-2}V_01_{\Omega}u^i$ satisfying the Sommerfield radiation conditions. This last equation can be written, in the semiclassical setting, as 
 $(-h^2\Delta -\kappa^2h^2 +V_01_{\Omega})u^s=-V_01_{\Omega}u^i$. We set $Res(-h^2\Delta +V_01_{\Omega}; k^2)$ to be the semiclasscial resolvent of $-h^2\Delta+V_01_{\Omega}$ at the frequency $k^2$. 
 The estimates discussed above imply, in particular, that $\Vert 1_{\Omega}Res(-h^2\Delta+V_01_{\Omega}; k^2)1_{\Omega}\Vert_{L^2(\mathbb{R}^3)}=O(h^{-1})$ at low frequencies, precisely for $k =O(h),\; h<<1$. 
 Hence, with the method described above, we can derive the semiclassical resolvent estimates at low frequencies, i.e. $\kappa =O(h),\; h<<1$, for positive, compactly supported but not necessarily 
   smooth potentials. The semiclasscial resolvent estimates for frequencies $k$ away from zero, and under non trapping conditions, are well known for general, but smooth enough, potentials, see for instance 
   \cite{C-P-V, R-T, V-Z} and the reference therein.

\bigskip

The rest of the paper is organized as follows. In Section \ref{Main-Theorem}, we prove Theorem \ref{C-Th} by using the expansions of Proposition \ref{Prpopostion-Expansion-Far-fields}.
The section is divided into two subsections, subsection \ref{main-section-P1}  and subsection \ref{Step-2:Semiclassical-estimates} corresponding to the estimates of $U^\infty(\hat{x},\theta)-U^\infty_{a}(\hat{x},\theta)$ and $U_{a}^\infty(\hat{x},\theta)-U^\infty_{D}(\hat{x},\theta)$ respectively.
The justification of estimate (\ref{key-estimate}) is provided in subsection \ref{Step-2:Semiclassical-estimates}.
The proof of Proposition \ref{Prpopostion-Expansion-Far-fields} is postponed to Section \ref{prop-Expansion-Far-fields}.

\section{Proof of Thereom \ref{C-Th}}\label{Main-Theorem}

The existence and uniqueness of the solution of the following scattering problem is guaranteed based on the Fredholm alternative, as it is discussed in section 2.2. 
Let $U^t_a$ be its unique solution.
 
 \begin{equation}\label{Umbounded-potentials-Th}
(\Delta + \kappa^{2}-a^{1-s}\bold{C}_0)U_{a}^{t}=0 \mbox{ in }\mathbb{R}^{3},
\end{equation}
\begin{equation}\label{Umbounded-potentials-BC}
U_{a}^{t}=U_{a}^s +e^{i\kappa x\cdot \theta}  
\end{equation}
\begin{equation}
\frac{\partial U_{a}^{s}}{\partial |x|}-i\kappa U_{a}^{s}=o\left(\frac{1}{|x|}\right), |x|\rightarrow\infty. \label{Umbounded-potentials-RC}
\end{equation}

\bigskip

Theorem \ref{C-Th} can split into the following two propositions

\begin{proposition}\label{P1}
Under the conditions on the distribution of the small holes
  \begin{eqnarray}\label{conditions-P1-1} 
     \frac{s}{3}\leq{t}<{1}&\mbox{ and } 0\leq{s}<\min\left\{(\frac{33}{29})_-,{2-t}, {3-3t},\frac{1}{2}(3-t),\frac{4-t}{3},\frac{4-3t}{2}\right\}
  \end{eqnarray}
  we have the following expansion:
  \begin{equation}\label{P1-1}
  U^\infty(\hat{x},\theta)- U_{a}^\infty(\hat{x},\theta)=O\left({ a^{\frac{(33-29s)_-}{12}}}+{a^{4-3s-t}}+{a^{4-2s-3t}}\right)~~ \mbox{ as } a \rightarrow 0,
 \end{equation}
 uniformly in terms of  $\theta$  and  $\hat{x} \mbox{ in } \mathbb{S}^2$.
\end{proposition}

\begin{proposition}\label{P2} Assume that $\Omega$ is a bounded domain with $C^{1,1}$ regularity and $s>1$, then we have the following expansion
\begin{equation}\label{P2-1}
  U_a^\infty(\hat{x},\theta)- U_{D}^\infty(\hat{x},\theta)=O(a^{\frac{s-1}{4}}), ~~ \mbox{ as } a \rightarrow 0,
 \end{equation}
 uniformly in terms of  $\theta$  and  $\hat{x} \mbox{ in } \mathbb{S}^2$.
\end{proposition}

\subsection{Proof of Proposition \ref{P1}}\label{main-section-P1}
 We rewrite the algebraic system \eqref{fracqcfracmain} in form

\begin{eqnarray}\label{fracqcfracmain-effect}
 Y_m +\sum_{\substack{j=1 \\ j\neq m}}^{M} \Phi_\kappa(z_m,z_j)\bar{C}_j Y_j a&=&U^{i}(z_m, \theta),
\end{eqnarray}
with $Y_m:=-\frac{Q_m}{C_m}$ and $C_m:=\bar{C}_m a$, for $m=1,\dots,M$ and $\bar{C}_m$ are related to the capacitances of $B_m$'s, i.e. 
they are independent of $a$. We set $\hat{C}:=(\bar{C}_1,\bar{C}_2,\dots,\bar{C}_M)^\top$ 
and define $ {\hat{Y}}:=\left(
    Y_1, Y_2 , \ldots  , Y_M \right)^\top \text{ and } 
\mathrm{U}^I:=\left(
     U^i(z_1), U^i(z_2), \ldots, U^i(z_M)\right)^\top$.

The following lemma ensures the invertibility of the algebraic system (\ref{fracqcfracmain-effect}), see its proof in \cite{C-S:2014}:

\begin{lemma}\label{Mazyawrkthm-effect}
If $a<\frac{5\pi}{3}\frac{d}{\|\hat{C}\|}$ and $t:=\min\limits_{j\neq\,m,1\leq\,j,m\leq\,M}\cos(\kappa|z_m-z_j|) \geq 0$, then the matrix $\mathbf{B}$ 
is invertible and the solution vector $\hat{Y}$ of \eqref{fracqcfracmain-effect} satisfies the estimate
\begin{equation}\label{mazya-fnlinvert-small-ac-2-effect}
 \sum_{m=1}^{M}|Y_m|^{2}
\leq4\left(1-\frac{3ta}{5\pi\,d}\|\hat{C}\|\right)^{-2}\sum_{m=1}^{M}\left|U^i(z_m)\right|^2,
\end{equation}
and hence the estimate
\begin{equation}\label{mazya-fnlinvert-small-ac-3-effect}
\begin{split}
 \sum_{m=1}^{M}|Y_m|
\leq2\left(1-\frac{3ta}{5\pi\,d}\|\hat{C}\|\right)^{-1}M\max\limits_{1\leq m \leq M}\left|U^i(z_m)\right|.
\end{split}
\end{equation}
\end{lemma}

\bigskip\par Consider the Lippmann-Schwinger equation
\begin{eqnarray}\label{fracqcfracmain-effect-int}
 Y(z) + a^{1-s}\int_{\Omega} \Phi_\kappa(z,y) \bar{{C_0}} (y) Y(y) dy &=&U^{i}(z, \theta), z\in \Omega,
\end{eqnarray}
where $\bar{{C_0}}$ is a piecewise constant function such that ${\bar{{C_0}}}|_{\Omega_m}={\bar{{C}}_m}$ for all $m=1,\dots,M$ and vanishes outside $\Omega$. 
As we have assumed that obstacles have same capacitances, ${\bar{{C}}_m}$ is same for every $m$ and let us denote it by ${\bar{{C}}}$.

For $m=1,\dots,M$, the equation \eqref{fracqcfracmain-effect-int} can be rewritten as
\begin{eqnarray}\label{fracqcfracmain-effect-int-1}
 Y(z_m) + a^{1-s}\sum_{\substack{j=1 \\ j\neq m}}^{M} \Phi_\kappa(z_m,z_j) \bar{{C}} Y(z_j) a^s&=&U^{i}(z_m, \theta)
   +a^{1-s}\left[\sum_{\substack{j=1 \\ j\neq m}}^{M} \Phi_\kappa(z_m,z_j) \bar{C}  Y(z_j) a^s-\int_{\Omega} \Phi_\kappa(z_m,y) \bar{{C_0}}(y) Y(y) dy\right]\nonumber\\
   &=&U^{i}(z_m, \theta)
   +a^{1-s}\left[\sum_{\substack{j=1 \\ j\neq m}}^{M} \Phi_\kappa(z_m,z_j) \bar{C}  Y(z_j) a^s-\int_{\Omega} \Phi_\kappa(z_m,y) \bar{{C_0}}(y) Y(y) dy\right]\nonumber\\
   &=&U^{i}(z_m, \theta)
   +a^{1-s}\bar{C}\left[\sum_{\substack{j=1 \\ j\neq m}}^{M} \Phi_\kappa(z_m,z_j)  Y(z_j) a^s-\sum_{j=1}^{M}\int_{\Omega_j} \Phi_\kappa(z_m,y)  Y(y) dy\right]\nonumber\\
   &&-a^{1-s}\int_{\Omega\setminus\cup_{j=1}^{M}\Omega_j} \Phi_\kappa(z_m,y)\bar{C}_0(y)  Y(y) dy\nonumber\\
      &=&U^{i}(z_m, \theta)
   +a^{1-s}\bar{C}\left[\underbrace{\sum_{\substack{j=1 \\ j\neq m}}^{M}\int_{\Omega_j} [\Phi_\kappa(z_m,z_j)  Y(z_j) - \Phi_\kappa(z_m,y)  Y(y)] dy}_{=: A}\right]\nonumber\\
   &&-a^{1-s}\bar{C}\underbrace{\int_{\Omega_m} \Phi_\kappa(z_m,y)  Y(y) dy}_{=:B}-a^{1-s}\underbrace{\int_{\Omega\setminus\cup_{j=1}^{M}\Omega_j} \Phi_\kappa(z_m,y)\bar{C}_0(y)  Y(y) dy}_{=:D}\nonumber\\
\end{eqnarray}

First, we prove the following Lemma.
\begin{lemma}\label{lem-Est-of-Y}
The function  $Y$ satisfying the Lippmann-Schwinger equation \eqref{fracqcfracmain-effect-int} satisfies the following estimates. 
\begin{eqnarray}\label{Eq-Est-of-Y}
\|Y\|_{L^\infty(\Omega)}\, =\,O(a^{\frac{1-s}{2}}),  &\qquad & \|\nabla{Y}\|_{L^\infty(\Omega)}\,=\,O(a^{\frac{3-\eta}{4}(1-s)}),
\end{eqnarray}
where $\eta$ is arbitrary positive quantity.
\end{lemma}

\begin{proof}
In Proposition \ref{Main-semi-classical-estimate-1}, we have the estimate $\Vert Y \Vert_{H^{\alpha}(\Omega)}=O(a^{(s-1)\frac{1-\alpha}{2}}), \alpha \in [0, 1], \; a<<1$.
For $\alpha=0$, we have $\Vert Y \Vert_{H^{\alpha}(\Omega)}=O(a^{\frac{s-1}{2}})$, and (\ref{fracqcfracmain-effect-int}) implies that 
\begin{equation}\label{Y-estimate}
 \vert Y(z)\vert = O(1) +O(a^{1-s}) O(a^{\frac{s-1}{2}})=O(a^{\frac{1-s}{2}}), 
\end{equation}
for $s \geq 1$.

Again, from (\ref{fracqcfracmain-effect-int}), we deduce that
$$
\vert \nabla Y(z)\vert =O(1) +a^{1-s} \Vert C(y) \nabla_z \Phi_{\kappa}(z, y)\Vert_{L^p(\Omega)} \Vert Y\Vert_{L^{p'}(\Omega)} 
$$
where $\frac{1}{p} +\frac{1}{p'}=1$. As for $p<\frac{3}{2}$, we have $\Vert \nabla_z \Phi_{\kappa}(z, y)\Vert_{L^p(\Omega)} <\infty$, we need only to estimate $\Vert Y\Vert_{L^{p'}(\Omega)}$ for $p'>3$.
We know also, by Sobolev embeddings, that $\Vert Y\Vert_{L^{p'}(\Omega)} \leq C \Vert Y\Vert_{H^{\alpha}(\Omega)}$ for $\frac{1}{p'}=\frac{1}{2}-\frac{\alpha}{3}$, i.e. $\alpha >\frac{1}{2}$. Then
$ \Vert Y\Vert_{L^{p'}(\Omega)} =O(a^{(s-1)\frac{1-\alpha}{2}}), \alpha >\frac{1}{2}$ and then
$$
 \vert \nabla Y(z)\vert =O(1)+O(a^{1-s +(s-1)\frac{1-\alpha}{2}})=O(a^{(1-s)\frac{1+\alpha}{2}}),\; \alpha >\frac{1}{2}, \mbox{ or }
$$
 
 \begin{equation}\label{nabla-Y-estimate} 
\|\nabla{Y}\|_{L^\infty(\Omega)}\,=\,O(a^{\frac{3-\eta}{4}(1-s)}) 
\end{equation}
where $\eta$ is arbitrary positive quantity.
\end{proof}
\bigskip

Based on Lemma \ref{lem-Est-of-Y}, we derive the following estimates of $A$,$B$ and $D$.
\begin{lemma}\label{Estimate-A-B} 
The quantities $A$,$B$ and $C$ enjoy the following estimates
\begin{eqnarray}
\vert A \vert\, =\,O(a^{\frac{3-s}{6}}+a^{\frac{(9-3\eta)-(5-3\eta)s}{12}}),   \vert B\vert \,=\,O(a^\frac{3+s}{6}) \mbox{ and } \vert D\vert \,=\,O(a^\frac{3-s}{6})
\end{eqnarray}
where $\eta$ is arbitrary positive quantity.
\end{lemma}

\begin{proof}
To evaluate A and D for every $z_m$, we need to perform the sum of the integrals when the $\Omega_j$'s are located inside $\Omega$ (for A) and when they are intersecting $\partial \Omega$ 
(for D). We proceed as follows in counting these $\Omega_j$'s:
\begin{itemize}

\item Evaluation of A. We start by distinguishing between the near and far-by obstacles to each obstacle. 
Let us suppose that these cubes are arranged in a cuboid, for example Rubik's cube, 
in different layers such that the total cubes upto the $n^{th}$
layer consists $(2n+1)^3$ cubes for $n=0,\dots,[a^{-\frac{s}{3}}]$, and $\Omega_m$ is located on the center. Hence the number of obstacles 
located in the $n^{th}$, $n\neq0$ layer  will have almost $K_{\max}[(2n+1)^3-(2n-1)^3]$ elements and their distance from $D_m$ is more than ${n}\left(\left(\frac{[K(z_m)+1]}{K(z_m)+1}\right)^\frac{1}{3}a^\frac{s}{3}-\frac{a}{2}\right)$. 

\item Evaluation of D. The corresponding $\Omega_j$'s are  touching the surface $\partial \Omega$. 
We distinguish two situations. 
\begin{enumerate} 
 \item In the first situation, the point $z_m$ is away from $\partial \Omega$ so that $\Phi_{\kappa}(z_m, y)$ is bounded in $y$ on $\partial \Omega$.
 \item In the second situation, the point $z_m$ is close to $\partial \Omega$, i.e. when $z_m$ is located near one of the $\Omega_j$'s touching $\partial \Omega$. 
Since the radius $a$ is small, and hence the $\Omega_j$'s, the point $z_m$ is close to $\partial \Omega$. As the function $\Phi_{\kappa}(z_m, y)$ is singular only for $y$ near $z_m$, we split the estimate of D into two parts.
One part involves the $\Omega_j$'s close to $z_m$,that we denote by $N_m$, and the other part involves the reminder, that we denote by $F_m$. The latter part can be estimated as in the first situation (i.e. 1. above). 
To estimate the former part (which is the worst part), we observe that, as $a$ is small enough, the $\Omega_j$'s close to $z_m$ are located near a small part of $\partial \Omega$ that we can assume to be flat, recalling that
$\partial \Omega$  is smooth. 
The point $z_m$ being close to this flat part, we divide this part into concentric layers as in the case when we evaluated A. But now, up to the $n^{th}$ layer at most $(2n+1)^2$ cubes are intersecting the surface,
 for  $n=0,\dots,[a^{-\frac{s}{3}}]$. As a consequence, the number  of obstacles 
located in the $n^{th}$, $n\neq0$ layer  will have at most $K_{\max}[(2n+1)^2-(2n-1)^2]$ elements and their distance from $D_m$ is more than 
${n}\left(\left(\frac{[K(z_m)+1]}{K(z_m)+1}\right)^\frac{1}{3}a^\frac{s}{3}-\frac{a}{2}\right)$. 

\end{enumerate}

\end{itemize}

Let us set $f(z_m,y):=\Phi_\kappa(z_m,y) Y(y)$. Using Taylor series, we can write
  $$f(z_m,y)-f(z_m,z_l)=(y-z_l)R_l(z_m,y),$$
  with 
  \begin{eqnarray}\label{taylorremind1}
   R_l(z_m,y)
   &=&\int_0^1\nabla_y f(z_m,y-\beta(y-z_l))\,d\beta\nonumber\\
   &=&\int_0^1\left[\nabla_y\Phi_\kappa(z_m,y-\beta(y-z_l))\right] Y(y-\beta(y-z_l))\,d\beta\nonumber\\
   &&+\int_0^1\Phi_\kappa(z_m,y-\beta(y-z_l))\left[\nabla_y Y(y-\beta(y-z_l))\right]\,d\beta.
  \end{eqnarray}
  From the explicit form of $\Phi_\kappa$, we have $\nabla_y\Phi_\kappa(x,y)=\Phi_\kappa(x,y)\left[\frac{1}{|x-y|}-i\kappa\right]\frac{x-y}{|x-y|}, {x}\neq{y}$.   
  Hence, we obtain that
  \begin{itemize}
   \item for $l\neq m$ , we have
   \begin{eqnarray*}
   \vert\Phi_\kappa(z_m,y-\beta(y-z_l))\vert\leq\frac{1}{4\pi{n\left(a^\frac{s}{3}-\frac{a}{2}\right)}},& \mbox{ and }&\vert\nabla_y\Phi_\kappa(z_m,y-\beta(y-z_l))\vert\leq\frac{1}{4\pi{n\left(a^\frac{s}{3}-\frac{a}{2}\right)}}\left[\frac{1}{{n\left(a^\frac{s}{3}-\frac{a}{2}\right)}}+\kappa\right].
  \end{eqnarray*}
 These values give us
 \begin{eqnarray}\label{taylorremind-effect-F}
 \vert R_l(z_m,y) \vert
   &\leq&\frac{1}{{n\left(a^\frac{s}{3}-\frac{a}{2}\right)}}\left(\left[\frac{1}{{n\left(a^\frac{s}{3}-\frac{a}{2}\right)}}+\kappa\right]\int_0^1 {\vert Y(y-\beta(y-z_l))\vert}d\beta+\int_0^1{\vert\nabla_y Y(y-\beta(y-z_l))\vert}d\beta\right).\nonumber\\
   &\leq&\frac{c_1}{{n\left(a^\frac{s}{3}-\frac{a}{2}\right)}}\left(\left[\frac{1}{{n\left(a^\frac{s}{3}-\frac{a}{2}\right)}}+\kappa\right]\|{Y}\|_{L^\infty(\Omega)}+c_5\|\nabla{Y}\|_{L^\infty(\Omega)}\right).
  \end{eqnarray}
  
  Hence, for $l\neq m$ using \eqref{taylorremind-effect-F} we get the following estimate for $A$;
  \begin{eqnarray}\label{integralonomega-subelements-abs}
|A|&=&\left\vert{\sum_{\substack{j=1 \\ j\neq m}}^{M}\int_{\Omega_j} [\Phi_\kappa(z_m,z_j)  Y(z_j) - \Phi_\kappa(z_m,y)  Y(y)] dy} \right\vert\nonumber \\
&\leq&{\sum_{\substack{j=1 \\ j\neq m}}^{[a^{-s}]}\left\vert\int_{\Omega_j} [\Phi_\kappa(z_m,z_j)  Y(z_j) - \Phi_\kappa(z_m,y)  Y(y)] dy\right\vert} ~~~ (\mbox{ as } M \leq [a^{-s}]  ~)\nonumber \\
&\leq&\sum_{n=1}^{[a^{-\frac{s}{3}}]}[(2n+1)^3-(2n-1)^3]a^\frac{s}{3}a^s\frac{c_1}{{n\left(a^\frac{s}{3}-\frac{a}{2}\right)}}\left(\left[\frac{1}{{n \left(a^\frac{s}{3}-\frac{a}{2}\right)}}+\kappa\right]\|{Y}\|_{L^\infty(\Omega)}+c_5\|\nabla{Y}\|_{L^\infty(\Omega)}\right) \nonumber \\
&=&O\left(\sum_{n=1}^{[a^{-\frac{s}{3}}]}[24n^2+2]a^{\frac{4}{3}s}\left[\frac{1}{n^2}a^{-\frac{2}{3}s}\|{Y}\|_{L^\infty(\Omega)}+\frac{1}{n}a^{-\frac{1}{3}s}\|\nabla{Y}\|_{L^\infty(\Omega)}\right]\right)\nonumber\\
&=&O\left(\sum_{n=1}^{[a^{-\frac{s}{3}}]}[24n^2+2]a^{\frac{4}{3}s}\left[\frac{1}{n^2}a^{-\frac{2}{3}s}a^{\frac{1-s}{2}}+\frac{1}{n}a^{-\frac{1}{3}s}a^{\frac{3-\eta}{4}(1-s)}\right]\right)\nonumber\\
&=&O\left(a^{\frac{3+s}{6}}\sum_{n=1}^{[a^{-\frac{s}{3}}]}[24+\frac{2}{n^2}]+a^{\frac{(1+\eta)s+(3-\eta)}{4}}\sum_{n=1}^{[a^{-\frac{s}{3}}]}[24n+\frac{2}{n}]\right)\nonumber\\
&=&O\left(a^{\frac{3+s}{6}}a^{-\frac{s}{3}}+a^{\frac{(1+\eta)s+(3-\eta)}{4}}a^{-\frac{2}{3}s}\right)\nonumber\\
&=&O(a^{\frac{3-s}{6}}+a^{\frac{(9-3\eta)-(5-3\eta)s}{12}}).
 \end{eqnarray}

 \item Let us estimate the integral value $\int_{\Omega_m} \Phi_\kappa(z_m,y) Y(y) dy$. We have the following estimates:
\begin{eqnarray}\label{estmatemthint-effe-acc}
 \left\vert\int_{\Omega_m} \Phi_\kappa(z_m,y)Y(y) dy\right\vert
 &\leq&\|Y\|_{L^\infty(\Omega)}\left\vert\int_{\Omega_m} \Phi_\kappa(z_m,y) dy\right\vert\nonumber\\
  &\leq&c_1a^{\frac{1-s}{2}}\left\vert\int_{\Omega_m} \Phi_\kappa(z_m,y) dy\right\vert\nonumber\\
 &\leq&\frac{1}{4\pi}c_1a^{\frac{1-s}{2}}\left(\int_{B(z_m,r)} \frac{1}{|z_m-y|} dy+\int_{\Omega_m\setminus B(z_m,r)} \frac{1}{|z_m-y|} dy\right)\nonumber\\
 &&{(\frac{1}{|z_m-y|} \in L^1(B(z_m,r)), r<\frac{1}{2}a^{\frac{s}{3}}) }\nonumber\\
 &\leq&\frac{1}{4\pi}c_1a^{\frac{1-s}{2}}\left(\sigma(\mathbb{S}^{3-1})\int_0^r \frac{1}{s}s^{3-1}ds+\frac{1}{r}Vol(\Omega_m\setminus B(z_m,r))\right)\nonumber\\
 &=&\frac{1}{4\pi}c_1a^{\frac{1-s}{2}}\underbrace{\left(2\pi r^2 +\frac{1}{r}\left[a^s-\frac{4}{3}\pi r^3\right]\right)}_{=:lm(r,a)}\nonumber\\
 &\leq&\frac{1}{4\pi}c_1a^{\frac{1-s}{2}}~lm(r^c,a),\nonumber\\
 &&\mbox{ $r^c$ is the value of $r$ where $lm(r,a)$ attains maximum}.\nonumber\\
 &&{\partial_r lm(r,a)=0 \Rightarrow 4\pi r-\frac{1}{r^2}a^s-\frac{8}{3}\pi r=0\Rightarrow  r_c=\left(\frac{3}{4}\pi a^s\right)^\frac{1}{3}}\nonumber\\
 &&{\begin{array}{ccc}
                    lm(r_c,a)&=&2\pi\left(\frac{3}{4}\pi\right)^\frac{2}{3} a^{\frac{2}{3}s}+\left(\frac{4}{3\pi} \right)^\frac{1}{3}a^{\frac{2}{3}s}-\frac{4}{3}\pi\left(\frac{3}{4}\pi \right)^\frac{2}{3}a^{\frac{2}{3}s}\\
                    &&\\
                    &=&\left[\frac{2}{3\pi}\left(\frac{3}{4}\pi\right)^\frac{2}{3}+\left(\frac{4}{3\pi} \right)^\frac{1}{3} \right]a^{\frac{2}{3}s}=\frac{3}{2}\left(\frac{4}{3\pi} \right)^\frac{1}{3}a^{\frac{2}{3}s}
                   \end{array}
}\nonumber\\ 
&=&\frac{3}{8\pi}c_1a^{\frac{1-s}{2}}\left(\frac{4}{3\pi} \right)^\frac{1}{3}a^{\frac{2}{3}s}=O(a^{\frac{3+s}{6}}).
\end{eqnarray}
\item Let us estimate $D$. Recall that $\vert \Omega \setminus \cup_{j=1}^{M}\Omega_j\vert$, and hence $\vert F_m\vert$, is of the order $a^{\frac{1}{3}s}$ as $a\rightarrow 0$.
 \begin{eqnarray}\label{integralonomega-subelements-abs-D}
|D|
&=&\left\vert\int_{\Omega\setminus\cup_{j=1}^{M}\Omega_j} \Phi_\kappa(z_m,y)\bar{C}_0(y)  Y(y) dy\right\vert\nonumber\\
&=&\left\vert\int_{N_m} \Phi_\kappa(z_m,y)\bar{C}_0(y)  Y(y) dy\right\vert +\left\vert\int_{F_m} \Phi_\kappa(z_m,y)\bar{C}_0(y)  Y(y) dy\right\vert \nonumber\\
&\leq&{\sum_{l=1}^{[a^{-\frac{2}{3}s}]}\|Y\|_{L^\infty(\Omega)}\|\bar{C}_0\|_{L^\infty(\Omega)}}\frac{1}{d^\prime_{ml}}|\Omega_l| + \|\Phi_\kappa(z_m, \cdot)\|_{L^\infty(F_m)} \|\bar{C}_0\|_{L^\infty(\Omega)} 
\|Y\|_{L^\infty(\Omega)} \vert F_m\vert \nonumber \\
&\leq&\|Y\|_{L^\infty(\Omega)}\|\bar{C}_0\|_{L^\infty(\Omega)}a^s{\sum_{l=1}^{[a^{-\frac{2}{3}s}]}}\frac{1}{d^\prime_{ml}} + C \|Y\|_{L^\infty(\Omega)} \vert F_m \vert\; ~~ (\mbox{ as } \Phi_\kappa (z_m, \cdot) \mbox{ is not singular in } F_m )\nonumber \\
&\leq&\|Y\|_{L^\infty(\Omega)}\|\bar{C}_0\|_{L^\infty(\Omega)}a^s{\sum_{l=1}^{[a^{-\frac{1}{3}s}]}}[(2n+1)^2-(2n-1)^2]\left(\frac{1}{n\left(2^{-\frac{1}{3}}a^\frac{s}{3}-\frac{a}{2}\right)}\right)+
C a^{\frac{1-s}{2}}a^{\frac{1}{3}s} \nonumber \\
&=&\|Y\|_{L^\infty(\Omega)}\|\bar{C}_0\|_{L^\infty(\Omega)}a^s O\left({\sum_{l=1}^{[a^{-\frac{1}{3}s}]}}[(2n+1)^2-(2n-1)^2]\left(\frac{1}{n\left(2^{-\frac{1}{3}}a^\frac{s}{3}-\frac{a}{2}\right)}\right)\right) + Ca^{\frac{3-s}{6}}\nonumber \\
&=&\|Y\|_{L^\infty(\Omega)}\|\bar{C}_0\|_{L^\infty(\Omega)}a^s O\left(a^{-\frac{2}{3}s}\right) + Ca^{\frac{3-s}{6}} \nonumber \\
&=\atop\eqref{Eq-Est-of-Y}&O\left(a^sa^{\frac{1-s}{2}}a^{-\frac{2}{3}s}\right)\quad=\quad O(a^{\frac{3-s}{6}}).
 \end{eqnarray}
\end{itemize}
\end{proof}

From these estimates of $A$, $B$ and $D$, we deduce that:

\begin{eqnarray}
Y(z_m) &+&\sum_{\substack{j=1 \\ j\neq m}}^{M} \Phi_\kappa(z_m,z_j)\bar{C}_j Y(z_j) a\\
  &=&U^{i}(z_m, \theta)+O(a^{\frac{3-s}{6}}+a^{\frac{(9-3\eta)-(5-3\eta)s}{12}})a^{1-s} +O\left(a^{\frac{3+s}{6}}\right)a^{1-s}+O\left(a^{\frac{3-s}{6}}\right)a^{1-s}\nonumber\\
   &=&U^{i}(z_m, \theta) +O\left(a^{\frac{9-5s}{6}}+a^\frac{9-7s}{6}+a^{\frac{(21-3\eta)-(17-3\eta)s}{12}}\right)\nonumber\\
    &=&U^{i}(z_m, \theta) +O\left(a^\frac{9-7s}{6}+a^{\frac{(21-3\eta)-(17-3\eta)s}{12}}\right).\label{fracqcfracmain-effect-int-3}
\end{eqnarray}

Taking the difference between \eqref{fracqcfracmain-effect} and \eqref{fracqcfracmain-effect-int-3} produce the algebraic system
\begin{eqnarray}\label{fracqcfracmain-effect-int-4}
 (Y_m-Y(z_m)) +\sum_{\substack{j=1 \\ j\neq m}}^{M} \Phi_\kappa(z_m,z_j)\bar{C}_j (Y_j-Y(z_j)) a 
 &=&O\left(a^\frac{9-7s}{6}+a^{\frac{(21-3\eta)-(17-3\eta)s}{12}}\right) .\nonumber
\end{eqnarray}

Comparing this system with \eqref{fracqcfracmain-effect} and by using Lemma \ref{Mazyawrkthm-effect}, we obtain the estimate

\begin{eqnarray}\label{mazya-fnlinvert-small-ac-3-effect-dif}
 \sum_{m=1}^{M}(Y_m-Y(z_m))&=&O\left(Ma^\frac{9-7s}{6}+Ma^{\frac{(21-3\eta)-(17-3\eta)s}{12}}\right).
\end{eqnarray}

Recalling that $d=a^t,\,M=O(a^{-s})$ with  $t,s>0$, we have the following approximation of the far-field from the Foldy-Lax asymptotic expansion \eqref{x oustdie1 D_m farmain-recent**} and 
from the definitions $Y_m:=-\frac{Q_m}{C_m}$ and $C_m:=\bar{C}_m a$, for $m=1,\dots,M$:
\begin{eqnarray}\label{x oustdie1 D_m farmain-recent**-effect}
U^\infty(\hat{x},\theta) &=&-\bar{C}~  \sum_{j=1}^{M}e^{-i\kappa\hat{x}\cdot z_j}Y_ja+O\left({a^{2-s}+a^{3-2s-t}+ a^{4-3s-t}}+{a^{3-s-2t}}+{a^{4-2s-3t}}\right).
\end{eqnarray}
Consider the far-field 
\begin{eqnarray}\label{acoustic-farfield-effect}
U^\infty_{a}(\hat{x},\theta) &=&-  a^{1-s}~\int_{\Omega} e^{-i\kappa\hat{x}\cdot{y}} \bar{C}_0(y)~ Y(y) dy.
\end{eqnarray}
Taking the difference between \eqref{acoustic-farfield-effect} and \eqref{x oustdie1 D_m farmain-recent**-effect} gives us;
\begin{eqnarray}\label{acoustic-difference-farfield-effect}
U^\infty(\hat{x},\theta)- U^\infty_{a}(\hat{x},\theta) &=& a^{1-s}\left[\int_{\Omega\setminus\cup_{j=1}^{M}\Omega_j} e^{-i\kappa\hat{x}\cdot{y}}\bar{C}_0 (y)Y(y) dy+\sum_{j=1}^{M} \int_{\Omega_j} e^{-i\kappa\hat{x}\cdot{y}}\bar{C}_0 (y)Y(y) dy- \sum_{j=1}^{M}e^{-i\kappa\hat{x}\cdot z_j}\bar{C}Y_ja^s \right]\nonumber\\ \nonumber
 &&
 +O\left({a^{2-s}+a^{3-2s-t}+ a^{4-3s-t}}+{a^{3-s-2t}}+{a^{4-2s-3t}}\right) \nonumber\\
&=& \bar{C} a^{1-s}\left[\sum_{j=1}^{M} \int_{\Omega_j} [e^{-i\kappa\hat{x}\cdot{y}}Y(y) - e^{-i\kappa\hat{x}\cdot z_j}Y(z_j)] \,dy\right]\nonumber\\ 
&&+\,\bar{C} a\left[\sum_{j=1}^{M}e^{-i\kappa\hat{x}\cdot{z_j}}[ Y(z_j)-Y_j] \right]+a^{1-s}\int_{\Omega\setminus\cup_{j=1}^{M}\Omega_j} e^{-i\kappa\hat{x}\cdot{y}}\bar{C}_0 (y)Y(y) dy\nonumber\\ 
 &&
 +O\left({a^{2-s}+a^{3-2s-t}+ a^{4-3s-t}}+{a^{3-s-2t}}+{a^{4-2s-3t}}\right) \nonumber\\
 &\substack{= \\ \eqref{mazya-fnlinvert-small-ac-3-effect-dif} }& 
 \bar{C} a^{1-s}\left[\sum_{j=1}^{M} \int_{\Omega_j} [e^{-i\kappa\hat{x}\cdot{y}}Y(y) - e^{-i\kappa\hat{x}\cdot z_j}Y(z_j)] \,dy\right]\nonumber\\
&&+\bar{C}\,a\,O\left(Ma^\frac{9-7s}{6}+Ma^{\frac{(21-3\eta)-(17-3\eta)s}{12}}\right)+O(a^{1-s}\|\bar{C}_0\|_{L^\infty{\Omega}}\|Y\|_{L^\infty{\Omega}}\vert\Omega\setminus\cup_{j=1}^{a^{-s}}\Omega_j \vert)\nonumber\\
   &&+O\left({a^{2-s}+a^{3-2s-t}+ a^{4-3s-t}}+{a^{3-s-2t}}+{a^{4-2s-3t}}\right)\nonumber\\
    &\substack{= \\ \eqref{Eq-Est-of-Y} }& 
 \bar{C} a^{1-s}\left[\sum_{j=1}^{M} \int_{\Omega_j} [e^{-i\kappa\hat{x}\cdot{y}}Y(y) - e^{-i\kappa\hat{x}\cdot z_j}Y(z_j)] \,dy\right]\nonumber\\
&&+\bar{C}\,a\,O\left(Ma^\frac{9-7s}{6}+Ma^{\frac{(21-3\eta)-(17-3\eta)s}{12}}\right)+O(a^{\frac{9-7s}{6}} )\nonumber\\
   &&+O\left({a^{2-s}+a^{3-2s-t}+ a^{4-3s-t}}+{a^{3-s-2t}}+{a^{4-2s-3t}}\right)
 \end{eqnarray}

By following the similar computations as it was done in (\ref{taylorremind1}-\ref{integralonomega-subelements-abs}), we can estimate the quantity 
'$\sum_{j=1}^{M}\int_{\Omega_j} \left[e^{-i\kappa\hat{x}\cdot{y}} Y(y) - e^{-i\kappa\hat{x}\cdot z_j}Y(z_j)\right]dy$'  by $O(a^{\frac{3-s}{6}}+a^{\frac{(9-3\eta)-(5-3\eta)s}{12}})$. Indeed,

 \begin{eqnarray}\label{integralonomega-subelements-exp-abs}
{\sum_{\substack{j=1 \\ j\neq m}}^{M}\left\vert\int_{\Omega_j} [e^{-i\kappa\hat{x}\cdot z_j}  Y(z_j) - e^{-i\kappa\hat{x}\cdot y} Y(y)] dy\right\vert} 
&\leq&\sum_{n=1}^{[a^{-\frac{s}{3}}]}[(2n+1)^3-(2n-1)^3]a^\frac{s}{3}a^s{c_1}\left(\kappa\|{Y}\|_{L^\infty(\Omega)}+c_5\|\nabla{Y}\|_{L^\infty(\Omega)}\right) \nonumber \\
&=&O\left(\sum_{n=1}^{[a^{-\frac{s}{3}}]}[24n^2+2]a^{\frac{4}{3}s}\left[\kappa\|{Y}\|_{L^\infty(\Omega)}+\|\nabla{Y}\|_{L^\infty(\Omega)}\right]\right)\nonumber\\
&=&O\left(\sum_{n=1}^{[a^{-\frac{s}{3}}]}[24n^2+2]a^{\frac{4}{3}s}\left[\kappa\,a^{\frac{1-s}{2}}+a^{\frac{3-\eta}{4}(1-s)}\right]\right)\nonumber\\
&=&O\left(a^{-{s}}a^{\frac{4}{3}s}\left[\kappa\,a^{\frac{1-s}{2}}+a^{\frac{3-\eta}{4}(1-s)}\right]\right)\nonumber\\
&=&O(a^{\frac{3-s}{6}}+a^{\frac{(9-3\eta)-(5-3\eta)s}{12}}).
 \end{eqnarray}

 We obtain
 \begin{eqnarray}\label{acoustic-difference-farfield-effect-1}
 U^\infty(\hat{x},\theta)-U^\infty_{a}(\hat{x},\theta) 
&=&O\Bigg({a^\frac{9-7s}{6}+a^{\frac{(21-3\eta)-(17-3\eta)s}{12}}}+{a^{\frac{15-13s}{6}}}+{ a^{\frac{(33-3\eta)-(29-3\eta)s}{12}}}\nonumber\\
&&\qquad +{a^{2-s}+a^{3-2s-t}+ a^{4-3s-t}}+{a^{3-s-2t}}+{a^{4-2s-3t}}\Bigg)\nonumber\\
    \end{eqnarray}
    
 Due to the fact that $d=a^t$ is the minimum distance between the small bodies,  we will have ${t\geq\frac{s}{3}}$. 
Otherwise, we will have contradiction as the volume of the collection of obstacles, which is of order $a^{-s}\left(\frac{a}{2}+\frac{d}{2}\right)^3$, explodes.
 From \eqref{conditions-1pr}, we also have that 
  \begin{eqnarray}\label{conditions-gem} 
     \frac{s}{3}\leq{t}<{1}&\mbox{ and } 0\leq{s}<\min\left\{2,{2-t}, {3-3t},\frac{1}{2}(3-t)\right\}.
  \end{eqnarray}
  We observe that for the particular case where $s=1$, we have 
    $$ U^\infty(\hat{x},\theta)-U^\infty_{a}(\hat{x},\theta) =O\left({ a^\frac{1}{3}}+{a^{1-t}}+{a^{2-3t}}\right)=O\left({ a^\frac{1}{3}}+{a^{2-3t}}\right),\,\frac{1}{3}\leq{t}<\frac{2}{3}.$$
  Under the general conditions
  \begin{eqnarray}\label{conditions-gem-1} 
     \frac{s}{3}\leq{t}<{1}&\mbox{ and } 0\leq{s}<\min\left\{{\frac{(33-3\eta)}{(29-3\eta)}},{2-t}, {3-3t},\frac{1}{2}(3-t),\frac{4-t}{3},\frac{4-3t}{2}\right\}
  \end{eqnarray}
  we have
   \begin{eqnarray}\label{acoustic-difference-farfield-effect-2}
 U^\infty(\hat{x},\theta)-U^\infty_{a}(\hat{x},\theta) 
&=&O\left({ a^{\frac{(33-3\eta)-(29-3\eta)s}{12}}}+{a^{4-3s-t}}+{a^{4-2s-3t}}\right)\nonumber\\
&=&\left\{\begin{array}{ccc}
           O\left({ a^{\frac{(33-3\eta)-(29-3\eta)s}{12}}}+{a^{4-2s-3t}}\right)&\mbox{ if }& s<2t;\\
             O\left({ a^{\frac{(33-3\eta)-(29-3\eta)s}{12}}}+{a^{\frac{8-7s}{2}}}\right)& \mbox{ if } &s=2t;\\
               O\left({ a^{\frac{(33-3\eta)-(29-3\eta)s}{12}}}+{a^{4-3s-t}}\right)& \mbox{ if }& s>2t;\\
          \end{array}
\right.\nonumber\\
    \end{eqnarray}
As $\eta$ can be taken arbitrary, we can rewrite this result as follows. Under the general conditions
  \begin{eqnarray}\label{conditions-gem-1-} 
     \frac{s}{3}\leq{t}<{1}&\mbox{ and } 0\leq{s}<\min\left\{(\frac{33}{29})_-,{2-t}, {3-3t},\frac{1}{2}(3-t),\frac{4-t}{3},\frac{4-3t}{2}\right\}
  \end{eqnarray}
  we have
   \begin{eqnarray}\label{acoustic-difference-farfield-effect-2-}
 U^\infty(\hat{x},\theta)-U^\infty_{a}(\hat{x},\theta) 
&=&O\left({ a^{\frac{(33-29s)_-}{12}}}+{a^{4-3s-t}}+{a^{4-2s-3t}}\right). \nonumber\\
    \end{eqnarray}
    For $\beta$, a real and positive number, we used the notation $\beta_-$ to describe the property $\alpha < \beta_-$if $\alpha  \leq \beta \; - r $ for some small $r>0$.

\subsection{Proof of Proposition \ref{P2}} \label{Step-2:Semiclassical-estimates}
\bigskip

\subsubsection{Reduction to a semiclassical type estimate}

Let us set $h:=a^{\frac{s-1}{2}}$, $V_0:=\bold{C}_0$, $u^i:=u^i\left(  x,\kappa\right):=e^{i\kappa\left(  x\cdot d\right)  }$ and $u^t:=U^t_a$, then the scattering problem 
(\ref{Umbounded-potentials-Th})-(\ref{Umbounded-potentials-BC})-(\ref{Umbounded-potentials-RC}) reduces to

 \begin{equation}\label{Umbounded-potentials}
(\Delta + \kappa^{2}-h^{-2}V_0)u^t=0 \mbox{ in }\mathbb{R}^{3},
\end{equation}
\begin{equation}\label{total-field}
u^t=u^s +e^{i\kappa x\cdot \theta}  
\end{equation}
\begin{equation}
\frac{\partial u^t_{s}}{\partial |x|}-i\kappa u^{s}=o\left(\frac{1}{|x|}\right), |x|\rightarrow\infty. \label{radiationc-D-1}
\end{equation}
 The solution of this scattering problem satisfies the following Lipmann-Schwinger integral equation%
\begin{equation}
u^t\left(  \cdot,\kappa\right)  =u^i\left(  \cdot,\kappa\right)  -\frac
{1}{h^{2}}\int_{\Omega}\Phi_{\kappa}\left(  \cdot,y\right)  V_{0}\,u^t\left(
y,\kappa\right)  \,dy\,. \label{Scatt_field_eq_Int}%
\end{equation}
Using
(\ref{Scatt_field_eq_Int}) and the large-$x$ behavior of $\Phi_{\kappa}\left(
x,y\right)  $, the radiation condition%
\begin{equation}
\left\vert x\right\vert \left(  \hat{x}\cdot\nabla-i\kappa\right)  u^s%
=o\left(  \frac{1}{\vert x\vert}\right)  \,. \label{radiation}%
\end{equation}
follows for the scattered field $u^s=u^t-u^i$. Convertly, any solution of (\ref{Scatt_field_eq_Int}) is a solution of  (\ref{Umbounded-potentials})-(\ref{total-field})-(\ref{radiationc-D-1}).
As the potential $h^{-2}V_0$ is real valued, then standard arguments based on the Fredholm alternative show that (\ref{Umbounded-potentials})-(\ref{total-field})-(\ref{radiationc-D-1}) has one and 
only one solution which is in $H^2_{loc}(\mathbb{R}^3)$, see \cite{C-K:1998} for instance.
\bigskip

The proof of Proposition \ref{P2} is reduced to the proof of the following property:

\begin{proposition}\label{Main-semi-classical-estimate-1}
Let
$\partial\Omega$ be of class $\mathcal{C}^{1,1}$. We have 
the following estimate%
\begin{equation}
\left\Vert \left.  u^t  \right\vert _{\partial
\Omega}\right\Vert _{H^{t}\left(  \partial\Omega\right)  }\leq K%
\,h^{1/2-t}\,,\quad t\in\left[  0,1/2\right]  \,,\label{rate_est_boundary-1}%
\end{equation}
with $K>0$.
\end{proposition}

Indeed, assuming that Proposition \ref{Main-semi-classical-estimate-1} is valid and taking $t=0$ in (\ref{rate_est_boundary-1}), we have the estimate $\left\Vert \left.  u^t  \right\vert _{\partial
\Omega}\right\Vert _{L^2\left(  \partial\Omega\right)}=O(h^{\frac{1}{2}})=O(a^{\frac{s-1}{4}})$, as $a<<1$. As $u^s$ satisfies $(\Delta +\kappa^2)u^s=0$ 
in $\mathbb{R}^3\setminus \bar{\Omega}$, the radiation conditions with $\left\Vert \left.  u^s\left(  x\right) + e^{i\kappa x\cdot \theta} \right\vert _{\partial
\Omega}\right\Vert _{L^2\left(  \partial\Omega\right)}=O(a^{\frac{s-1}{4}})$ as $a<<1$, then the well posedness of the forward scattering problem in the exterior domain $\mathbb{R}^3\setminus \bar{\Omega}$ implies that
the corresponding far-fields satisfy the estimate 
\begin{equation}\label{Main-estimate-far-fields}
\vert U^\infty(\hat{x}, \theta)-U^\infty_D(\hat{x}, \theta)\vert = O(a^{\frac{s-1}{4}}) \mbox{ as } a<<1.
\end{equation}

\subsubsection{Proof of Proposition \ref{Main-semi-classical-estimate-1}}

The starting point is the following Lippmann-Schwinger equation
\begin{equation}\label{Lipp-Schwinger-equation}
(h^2+R(\kappa))u^t=h^2u^{i}\;~~~ \mbox{ in } \Omega
\end{equation}
where 
\begin{equation}
R(\kappa) u^{t}(x):=\int_{\Omega}\Phi_\kappa(x, y)V_0u^t(y)dy.
\end{equation}
We write $R(\kappa):=R(0)+P(\kappa)$, where now

\begin{equation}
R(0)u^t:=\int_{\Omega}\Phi_0(x, y)V_0u^t(y)dy
~~~ \mbox{ and }~~~ P(\kappa)u^t(y):=\int_{\Omega}(\Phi_\kappa(x, y)-\Phi_0(x, y))V_0u^t(y)dy 
\end{equation}
with $\Phi_0(x, y):=(4\pi \vert x-y\vert)^{-1}$.

We rewrite the equation (\ref{Lipp-Schwinger-equation}) as:
\begin{equation}
(h^2+R(0))u^t=h^2u^{i}-P(\kappa)u^t.
\end{equation}
Let $u^t_1$ be defined as the solution of the equation
\begin{equation}
(h^2+R(0))u_1^t=h^2u^{i}
\end{equation}
and $u^t_2$ be the one of the equation
\begin{equation}
(h^2+R(0))u_2^t=-P(\kappa)u^t.
\end{equation}
It is clear that $u^t=u_1^t+u_2^t$. 
The estimates of both $u^t_1$ and $u^t_2$ are based on the following properties of the operator $R(0)$.
\begin{proposition}\label{Prop-spectral}
Let $(\lambda_n, e_n)$ be the sequence of the eigenelements of the self-adjoint and compact operator $R(0): L^2(\Omega) \longrightarrow L^2(\Omega)$. We have the characterization
\begin{equation}
u\in H^1(\Omega) \Longleftrightarrow \sum_{n}\lambda^{-1}_n(u, e_n)^2_{L^2(\Omega)} <\infty
\end{equation} 
and there exist two  positive constants $c(\Omega)$ and $C(\Omega)$ depending only on $\Omega$ such that for every $u \in H^1(\Omega)$
\begin{equation}\label{H1-norm-equivalence}
c(\Omega) \Vert u \Vert_{H^1(\Omega)} \leq (\sum_{n}\lambda^{-1}_n V_0 (u, e_n)^2_{L^2(\Omega)})^{\frac{1}{2}} \leq C(\Omega) \Vert u \Vert_{H^1(\Omega)} 
\end{equation}
and 
\begin{equation}\label{H1-dual-norm-equivalence}
(C(\Omega))^{-1} \Vert u \Vert_{(H^1(\Omega))'} \leq (\sum_{n}\lambda_n V^{-1}_0 (u, e_n)^2_{L^2(\Omega)})^{\frac{1}{2}} \leq (c(\Omega))^{-1} \Vert u \Vert_{(H^1(\Omega))'} 
\end{equation}
where we used the scalar product $(u, v)_{L^2}:=\int_{\Omega}u(x)v(x)dx$.
\end{proposition}

\begin{proof}
We observe that if $(\lambda_n, e_n)_{n \in \mathbb{N}}$ are the eigenvalues and eigenfunctions of the (Newtonian potential) operator $R(0)$ then $(\lambda^{-1}_nV_0, e_n)_{n \in \mathbb{N}}$ are exactly the ones of the problem
\begin{equation}\label{Operator A0}
-\Delta e_n=\lambda_n^{-1} V_0 e_n, \mbox{ in } \Omega \mbox{ and } \partial_\nu e_n=S^{-1}(-\frac{1}{2}I+K)e_n, \mbox{ on } \partial \Omega
\end{equation}
where $S$ is the single layer operator $Sf(x):=\int_{\partial \Omega} \Phi_0(x, y) f(y)ds(y)$ and $K$ the double layer operator $Kf(x):=\int_{\partial \Omega}\partial_{\nu(y)} \Phi_0(x, y)f(y)ds(y)$
with $\nu$ as the exterior unit normal to $\Omega$.
We recall that the operators $S$ and $K$ are well defined, linear and bounded from $H^{-\frac{1}{2}}(\partial \Omega)$ to $H^{\frac{1}{2}}(\partial \Omega)$ and from $H^{\frac{1}{2}}(\partial \Omega)$ to itself respectively.
In addition $S$ is invertible. This observation has been made since at least Kac \cite{Kac}. More details are provided in \cite{KS, RRS}. 
The reader can get the boundary condition above by simply writing the Green's formula for $\Phi_0$ and $e_n$ inside $\Omega$ and then taking the trace on $\partial \Omega$ using the jumps of the double layer operator.

We set $B:=S^{-1}(-\frac{1}{2}I+K)$. We define the operator
$(A_0, \mathcal{D}(A_0))$ as $\mathcal{D}(A_0):=\{u\in H^1(\Omega), \Delta u \in L^2(\Omega) \mbox{ and } \partial_\nu u =B\;u\}$ and  $A_0u=-\Delta u$ for $u\in \mathcal{D}(A_0)$. 
We set the corresponding quadratic form 
\begin{equation}
a_0(u, v):=\int_{\Omega}\nabla u (x)\cdot \nabla v (x) dx -\int_{\partial \Omega} Bu (x)\; v(x)ds(x).
\end{equation}
This form is well defined in $H^1(\Omega)\times H^1(\Omega)$. In addition, it is symmetric. Indeed, we have
$$
\int_{\partial \Omega} Bu (x)\; v(x)ds(x)=\int_{\partial \Omega} S^{-1}(-\frac{1}{2}I+K)u (x)\; v(x)ds(x)=\int_{\partial \Omega}(-\frac{1}{2}I+K)u (x)\; S^{-1}v(x)ds(x).
$$ 
We write $u=S S^{-1} u$, then as $KS=SK^*$, see \cite{Costabel} for instance, where $K^*$ is the dual of double layer potential $K$, we derive
$$
\int_{\partial \Omega}(-\frac{1}{2}I+K)u (x)\; S^{-1}v(x)ds(x)=\int_{\partial \Omega}S(-\frac{1}{2}I+K^*) S^{-1}u (x)\; S^{-1}v(x)ds(x) =\int_{\partial \Omega}u(x)S^{-1}(-\frac{1}{2}I+K)v (x)ds(x).
$$
Hence $ \int_{\partial \Omega} Bu (x)\; v(x)ds(x)=\int_{\partial \Omega} u(x)\; Bv (x)ds(x)$
and then $a_0(u, v)=a_0(v, u), \mbox{ for } u, v \in H^1(\Omega)$.
\bigskip

Now, we show that $a_0(u, u)\geq \int_{\Omega}\vert \nabla\; u\vert^2$. Indeed,
$$
-\int_{\partial \Omega} Bu (x)\; u(x)ds(x)=\int_{\partial \Omega} S^{-1}(\frac{1}{2}I-K)u (x)\; u(x)ds(x)=\int_{\partial \Omega} u\;(\frac{1}{2}I-K^*)S^{-1} u(x)ds(x).
$$
As the operator $\frac{1}{2}I-K^*$ is positive definite on $H^{-\frac{1}{2}}(\partial \Omega)$ equipped with scalar product $<Su, v>$, see \cite{Costabel} for instance, we obtain
$$
-\int_{\partial \Omega} Bu (x)\; u(x)ds(x)=\int_{\partial \Omega} S (S^{-1}u)(x) \;(\frac{1}{2}I-K^*)S^{-1} u(x)ds(x) \geq c_0 \Vert S^{-1}u\Vert_{H^{-\frac{1}{2}}(\partial \Omega)}
$$
with a positive constant $c_0$. Hence $a_0(u, u)\geq \int_{\Omega}\vert \nabla\; u\vert^2$.
\bigskip

Let us define the bilinear form $a_{\alpha}(u, v):=a_0(u, v)+\alpha (u, v)_{L^2(\Omega)}$ for $u, v \in H^1(\Omega)$, with a positive constant $\alpha$. This form is defined on $H^1(\Omega)$, symmetric, continuous and coercitive. 
To $a_\alpha$, corresponds a self-adjoint operator with a compact inverse $A_\alpha: L^2(\Omega)\longrightarrow L^2(\Omega)$. 
From the spectral theory, see \cite{DL}, we know that the eigenvalues and eigenfunctions $(\lambda^{\alpha}_n, e^{\alpha}_n)_{n \in \mathbb{N}}$ of $A_\alpha$ define a basis in $L^2(\Omega)$ and, in addition, we have 
$$
H^1(\Omega)=\mathcal{D}((A_\alpha)^{\frac{1}{2}}).
$$ 
This means that 
\begin{equation}\label{H1-cha-0}
u\in H^1(\Omega) \Longleftrightarrow \sum_{n}\lambda^{\alpha}_n(u, e^{\alpha}_n)^2_{L^2(\Omega)} <\infty ~~ 
\mbox {and }
\Vert u \Vert^2_{H^1(\Omega)} \sim \sum_{n}\lambda^{\alpha}_n(u, e^{\alpha}_n)^2_{L^2(\Omega)}.
\end{equation}
The constants appearing in the two inequalities of this equivalence might depend on $\Omega$ but they are independent on $V_0$.
As $A_0$ is the operator corresponding to $a_0$, then we have 
$$
\lambda^{\alpha}_n=\lambda^{-1}_n V_0+\alpha \mbox{ and } e^{\alpha}_n=e_n, \mbox{ for every } n.
$$
Finally, as $0<\lambda_n$ for every $n$, $\lambda^{-1}_n \longrightarrow \infty$, as $n\longrightarrow \infty$, and $\alpha$ is arbitrary small, say $\alpha <1$, we deduce from (\ref{H1-cha-0}) that
$$
u\in H^1(\Omega) \Longleftrightarrow \sum_{n}\lambda^{-1}_nV_0 (u, e_n)^2_{L^2(\Omega)} <\infty
$$
and 
$$
c(\Omega)\Vert u \Vert_{H^1(\Omega)} \leq (\sum_{n}\lambda^{-1}_n V_0 (u, e_n)^2_{L^2(\Omega)})^{\frac{1}{2}} \leq C(\Omega) \Vert u \Vert_{H^1(\Omega)}.
$$
As mentioned above, the constants appearing in these inequalities are independent on $V_0$. 
\bigskip

Let us now deal with the $(H^1(\Omega))'$ norm. As we have $<u, f>_{(H^1(\Omega))', H^1(\Omega)}=(u, f)_{L^2(\Omega)}$ for $u \in L^2(\Omega)$ and $f\in H^1(\Omega)$, then 
$$
\Vert u \Vert_{(H^1(\Omega))'}:=\sup_{\Vert f\Vert_{H^1(\Omega)}=1}\vert <u, f>\vert= \sup_{\Vert f\Vert_{H^1(\Omega)}=1}\vert \sum_n(u, e_n)(f, e_n)\vert
=\sup_{\Vert f\Vert_{H^1(\Omega)}=1}\vert \sum_n \lambda_n^{\frac{1}{2}}V^{-\frac{1}{2}}_0(u, e_n)\lambda_n^{-\frac{1}{2}}V^{\frac{1}{2}}_0(f, e_n)\vert
$$
$$
\leq \sup_{\Vert f\Vert_{H^1(\Omega)}=1}\vert [\sum_n \lambda_nV^{-1}_0(u, e_n)^2]^{\frac{1}{2}}[\sum_n\lambda_n^{-1}V_0(f, e_n)^2]^{\frac{1}{2}}
$$
$$
\leq [\sum_n \lambda_nV^{-1}_0(u, e_n)^2]^{\frac{1}{2}} \sup_{\Vert f\Vert_{H^1(\Omega)}=1}\vert[\sum_n\lambda_n^{-1}V_0(f, e_n)^2]^{\frac{1}{2}}\leq C(\Omega)[\sum_n \lambda_nV^{-1}_0(u, e_n)^2]^{\frac{1}{2}}.
$$
Convertly, 
$$
\sum_n \lambda_nV^{-1}_0(u, e_n)^2=\sum_n \lambda_nV^{-1}_0(u, e_n)(u, e_n)=(u, \sum_n \lambda_nV^{-1}_0(u, e_n)e_n)
$$
Set $f_0:=\sum_n \lambda_nV^{-1}_0(u, e_n)e_n$, and let us show that it is in $H^1(\Omega)$. Indeed, $(f_0, e_n)=\lambda_nV^{-1}_0(u, e_n)$, then 
$$\lambda^{-1}_nV_0 (f_0, e_n)^2=\lambda_nV^{-1}_0(u, e_n)^2 <\lambda^{-1}_nV^{-1}_0(u, e_n)^2$$
since $\lambda_n \leq \lambda^{-1}_n$, for $n$ large enough.
As $u \in H^1(\Omega)$, then $\sum_n \lambda^{-1}_nV_0 (f_0, e_n)^2 < V_0^2 \sum \lambda_nV^{-1}_0(u, e_n)^2<\infty$.
This means that $f_0$ is in $H^1(\Omega)$. Hence we have
$$
\sum_n \lambda_nV^{-1}_0(u, e_n)^2 =(u, f_0)=\vert (u, f_0)\vert \leq \frac{\vert (u, f_0)\vert}{\Vert f_0\Vert_{H^1(\Omega)}}\;~ \Vert f_0\Vert_{H^1(\Omega)} \leq 
\sup_{f\in H^1(\Omega)}\frac{\vert (u, f)\vert}{\Vert f\Vert_{H^1(\Omega)}}
\;~ \Vert f_0\Vert_{H^1(\Omega)}.
$$
But $ \sup_{f\in H^1(\Omega)}\frac{\vert (u, f)\vert}{\Vert f\Vert_{H^1(\Omega)}} \leq \Vert u \Vert_{(H^1(\Omega))'}$ and from (\ref{H1-norm-equivalence}), we have
$$
\Vert f_0\Vert_{H^1(\Omega)} \leq (c(\Omega))^{-1}(\sum_n \lambda_n^{-1}V_0 (f, e_n)^2)^{\frac{1}{2}}=(c(\Omega))^{-1}(\sum_n \lambda_nV^{-1}_0 (u, e_n)^2)^{\frac{1}{2}}
$$
then we deduce that
$$
(\sum_n \lambda_nV^{-1}_0 (u, e_n)^2)^{\frac{1}{2}} \leq (c(\Omega))^{-1} \Vert u\Vert_{(H^1(\Omega))'}.
$$
\end{proof}

\underline{{\bf{Estimates of $u_1^t$}}}
\bigskip

We recall that 
\begin{equation}
(h^2+R(0))u_1^t=h^2u^{i}.
\end{equation}
By the spectral decomposition, we write $u_1^t=\sum_n (u_1^t, e_n)e_n$, then 
\begin{equation}\label{u-1-spectral-decomposition}
(u_1^t, e_n)=\frac{h^2}{h^2+\lambda_n}(u^{i}, e_n).
\end{equation}
 Multiplying both sides by $\lambda^{-\frac{1}{2}}_n$, and as $\lambda_n>0$, for every $n$, we get
$$
\sum_n \lambda_n^{-1} V_0 \vert ( u_1^t, e_n )\vert^2 \leq \sum_n \lambda_n^{-1} V_0 \vert ( u^i, e_n )\vert^2
$$
By Proposition \ref{Prop-spectral}, we deduce that
\begin{equation}\label{Est-u_1-1}
\Vert u_1^t\Vert_{H^1(\Omega)}\leq (c(\Omega))^{-1} C(\Omega) \Vert u^i\Vert_{H^1(\Omega)}.
\end{equation}
In addition, from (\ref{u-1-spectral-decomposition}), we have $\vert ( u_1^t, e_n )\vert \leq h^2 \lambda^{-1}_n \vert ( u^{i}, e_n)\vert $ 
and then
$$
V_0^2\sum _n\lambda_n V^{-1}_0 \vert ( u_1^t, e_n )\vert^2 \leq h^4 \sum_n\lambda^{-1}_nV_0 \vert (u^{i}, e_n)\vert^2.
$$
As $\Vert u_1^t\Vert^2_{(H^1(\Omega))'}\leq C^2(\Omega)\sum _n\lambda_n V^{-1}_0 \vert (u_1^t, e_n ) \vert^2$, 
 this means that
\begin{equation}\label{Est-u_1-2}
\Vert u_1^t\Vert_{(H^1(\Omega))'}\leq C(\Omega)^2 V^{-1}_0 h^2 \Vert u^i\Vert_{H^1(\Omega)}=O(V^{-1}_0 h^2), h<<1. 
\end{equation}

\underline{{\bf{Estimates of $u_2^t$}}}
\bigskip

We recall that 
\begin{equation}
(h^2+R(0))u_2^t=-P(\kappa)u^t.
\end{equation}
By the spectral decomposition, we get
\begin{equation}
(u^t_2, e_n)=-\frac{(P(\kappa)u^t, e_n)}{h^2+\lambda_n}.
\end{equation}
Hence $\vert (u^t_2, e_n)\vert \leq \lambda^{-1}_n\vert (P(\kappa)u^t, e_n)\vert $ or
$$
\lambda^{\frac{1}{2}}_n\vert (u^t_2, e_n)\vert \leq \lambda^{-\frac{1}{2}}_n\vert (P(\kappa)u^t, e_n)\vert
$$
But $P(\kappa)u^t=\sum_m(u^t, e_m)P(\kappa)e_m$ as $P(\kappa)$ is linear and bounded. Hence
$$
\lambda^{\frac{1}{2}}_n\vert (u^t_2, e_n)\vert \leq \lambda^{-\frac{1}{2}}_n \vert \sum_m(u^t, e_m)(P(\kappa)e_m, e_n)\vert
\leq (\sum _m\lambda_m(u^t, e_m)^2)^{\frac{1}{2}}(\sum_m \lambda_m^{-1} \lambda_n^{-1}(P(\kappa) e_m, e_n )^2)^{\frac{1}{2}}
$$
and then
\begin{equation}\label{u_2-P(kappa)}
V_0^2\sum_n \lambda_n V_0^{-1} \vert (u^t_2, e_n)\vert^2 
\leq (\sum _m\lambda_mV_0^{-1}(u^t, e_m)^2)\sum_n\sum_m \lambda_m^{-1}V_0 \lambda_n^{-1}V_0(P(\kappa) e_m, e_n )^2
\end{equation}
 Let us estimate the term $\sum_n\sum_m \lambda_m^{-1}V_0 \lambda_n^{-1}V_0(P(\kappa) e_m, e_n )^2$. First, as $P(k)$ is symmetric, we observe that
 $$
 \sum_m \lambda_m^{-1}V_0(P(\kappa) e_m, e_n )^2=\sum_m \lambda_m^{-1} V_0 ( P(\kappa)e_n, e_m )^2 \leq C^2(\Omega) \Vert P(\kappa)e_n\Vert^2_{H^1(\Omega)}.
 $$
 Then, we have
 $$
 \sum_n\sum_m \lambda_m^{-1}V_0 \lambda_n^{-1} V_0(P(\kappa) e_m, e_n )^2 \leq C^2(\Omega) \sum_n \lambda^{-1}_nV_0\Vert P(\kappa)e_n\Vert^2_{H^1(\Omega)}.
 $$
But 
$$
\sum_n \lambda^{-1}_nV_0\Vert P(\kappa)e_n\Vert^2_{H^1(\Omega)} =\sum_n \lambda^{-1}_nV_0\int_{\Omega}(\int_{\Omega}p_\kappa(x, y)e_n(y)dy)^2dx+\sum_n \lambda^{-1}_nV_0\int_{\Omega}(\int_{\Omega}\nabla_{x}p_\kappa(x, y)e_n(y)dy)^2dx
$$
recalling that $P(\kappa)e_n(x):=\int_{\Omega}p_\kappa(x, y)e_n(y)dy$ where $p_\kappa(x, y):=\Phi_\kappa(x, y)-\Phi_0(x, y)$.
Then by Lebesgue dominated Theorem, we get
$$
\sum_n \lambda^{-1}_nV_0\Vert P(\kappa)e_n\Vert^2_{H^1(\Omega)} =\int_{\Omega}\sum_n \lambda^{-1}_nV_0(\int_{\Omega}p_\kappa(x, y)e_n(y)dy)^2dx+\int_{\Omega}\sum_n \lambda^{-1}_nV_0(\int_{\Omega}\nabla_{x}p_\kappa(x, y)e_n(y)dy)^2dx
$$
$$
\leq C^2(\Omega) [\int_{\Omega}\Vert p_\kappa(x, \cdot)\Vert^2_{H^1(\Omega)}dx+
\int_{\Omega}\Vert \nabla_x p_\kappa(x, \cdot)\Vert^2_{H^1(\Omega)}dx] =: C^2(\Omega)C^2(\kappa).
$$

As $p_\kappa(x, y)=\Phi_\kappa(x, y)-\Phi_0(x, y)$, then it satisfies
$$
\Delta_y p_\kappa(x, y) =-\kappa^2 \Phi(x, y)
$$
and as $\Phi(x, \cdot)$ is in $L_{loc}^2(\mathbf{R}^3)$, then by interior estimates, we deduce that $p_\kappa(x, \cdot)$ is in $H^2_{loc}(\mathbf{R}^3)$ 
and $\Vert p_\kappa(x, \cdot)\Vert_{H^2(\Omega)}$ is uniformly bounded with respect to $x \in \Omega$. Finally as $\nabla_x p_\kappa(x, y)=-\nabla_y p\kappa(x, y)$, then $C(\kappa)$ make sense, i.e. its is finite.

Then from (\ref{u_2-P(kappa)}), we deduce that
\begin{equation}\label{H1-u-2-t}
\Vert u^t_2\Vert_{(H^1(\Omega))'} \leq \frac{(c(\Omega))^{-1}C^2(\Omega)C(\kappa)}{V_0} \Vert u^t\Vert_{(H^1(\Omega))'}.
\end{equation}

Let us now go back to the relation $ (u_2^t, e_n)=-\frac{(P(\kappa)u^{t}, e_n)}{h^2+\lambda_n}$
and derive the estimate $ \vert (u^t_2, e_n)\vert \leq \frac{\vert (P(\kappa)u^{t}, e_n)\vert}{h^2}$
or $ \lambda^{-\frac{1}{2}}_n\vert (u^t_2, e_n)\vert \leq \lambda^{-\frac{1}{2}}_n\frac{\vert (P(\kappa)u^{t}, e_n)\vert}{h^2}$ and then
$$
\sum_n \lambda^{-1}_nV^{-1}_0\vert (u^t_2, e_n)\vert \leq h^{-4}\sum_n\lambda^{-1}_nV_0^{-1}\vert (P(\kappa)u^{t}, e_n)\vert
$$
which means that
$$
\Vert u^t_2\Vert_{H^1(\Omega)}\leq (c(\Omega))^{-1} C(\Omega) h^{-2} \Vert P(\kappa)u^{t}\Vert_{H^1(\Omega)}.
$$
Let us now estimate $\Vert P(\kappa)u^{t}\Vert_{H^1(\Omega)}$. From the equality
$$
\Vert P(\kappa)u^{t}\Vert^2_{H^1(\Omega)} =\int_{\Omega}(\int_{\Omega}p_\kappa(x, y)u^t(y)dy)^2dx +\int_{\Omega}(\int_{\Omega}\nabla_x p_\kappa(x, y)u^t(y)dy)^2dx
$$
and the fact that
$$
\vert \int_{\Omega}p_\kappa(x, y)u^t(y)dy \vert \leq \Vert u^t\Vert_{(H^1(\Omega))'} \Vert p_\kappa(x, \cdot)\Vert_{H^1(\Omega)}
$$
and
$$
\vert \int_{\Omega}\nabla_x p_\kappa(x, y)u^t(y)dy \vert \leq \Vert u^t\Vert_{(H^1(\Omega))'} \Vert \nabla_x p_\kappa(x, \cdot)\Vert_{H^1(\Omega)}
$$
we have 
$$
\Vert P(\kappa)u^{t}\Vert^2_{H^1(\Omega)} \leq [\int_{\Omega}\Vert p_\kappa(x, \cdot)\Vert^2_{H^1(\Omega)}dx+\int_{\Omega}\Vert \nabla_x p_\kappa(x, \cdot)\Vert^2_{H^1(\Omega)}dx]\; \Vert u^t\Vert^2_{(H^1(\Omega))'}=C(\kappa)\Vert u^t\Vert^2_{(H^1(\Omega))'}.
$$
From the singularities of the function $p_\kappa (\cdot, \cdot)$, as described above, we know that the two quantities 
$\int_{\Omega}\Vert p_\kappa(x, \cdot)\Vert^2_{H^1(\Omega)}dx$ and $\int_{\Omega}\Vert \nabla_x p_\kappa(x, \cdot)\Vert^2_{H^1(\Omega)}dx$ are finite.

Finally, we have
$$
\Vert u^t_2\Vert_{H^1(\Omega)}\leq (c(\Omega))^{-1} C(\Omega) C(\kappa)h^{-2} \Vert u^t \Vert_{(H^1(\Omega))'}.
$$

\underline{{\bf{Estimates of $u^t$}}}
\bigskip

We have shown so far that
$$
\Vert u^t_1\Vert_{H^1(\Omega)}\leq (c(\Omega))^{-1} C(\Omega) \Vert u^i\Vert_{H^1(\Omega)}=O(1) \mbox{ and }
\Vert u^t_1\Vert_{(H^1(\Omega))'}\leq C^2(\Omega) V^{-1}_0h^2 \Vert u^i\Vert_{H^1(\Omega)}=O(V^{-1}_0h^2)
$$
$$
\Vert u^t_2\Vert_{H^1(\Omega)}\leq (c(\Omega))^{-1} C(\Omega)C(\kappa) h^{-2} \Vert u^t \Vert_{(H^1(\Omega))'}
\mbox{ and }
\Vert u^t_2\Vert_{(H^1(\Omega))'}\leq \frac{(c(\Omega))^{-1}C^2(\Omega)C(\kappa)}{V_0} \Vert u^t\Vert_{(H^1(\Omega))'}
$$
From these estimates, we deduce that
$$
\Vert u^t\Vert_{(H^1(\Omega))'}\leq \Vert u_1^t\Vert_{(H^1(\Omega))'} +\Vert u_2^t\Vert_{(H^1(\Omega))'}\leq C^2(\Omega)V^{-1}_0h^2 \Vert u^i \Vert_{H^1(\Omega)} +
\frac{(c(\Omega))^{-1} C^2(\Omega)C(\kappa)}{V_0} \Vert u^t\Vert_{(H^1(\Omega))'}.
$$
Hence under the condition on $\kappa$ 
\begin{equation}\label{condition-kappa}
C(\kappa)<\frac{c(\Omega)V_0}{C^2(\Omega)}
\end{equation}
we derive the estimate
$$
\Vert u^t\Vert_{(H^1(\Omega))'} \leq \frac{C^2(\Omega)V^{-1}_0 h^2 \Vert u^i \Vert_{H^1(\Omega)}}{1- \frac{C^2(\Omega)C(\kappa)}{V_0 c(\Omega)}}.
$$
Using this estimate, we derive the following one
$$
\Vert u^t\Vert_{H^1(\Omega)} \leq \Vert u_1^t\Vert_{H^1(\Omega)} +\Vert u_2^t\Vert_{H^1(\Omega)} \leq (c(\Omega))^{-1} C(\Omega)\Vert u^i\Vert_{H^1(\Omega)} +
(c(\Omega))^{-1}C(\Omega)C(\kappa) h^{-2}\frac{V^{-1}_0h^2 C^2(\Omega) \Vert u^i \Vert_{H^1(\Omega)}}{1- 
\frac{C^2(\Omega)C(\kappa)}{V_0 c(\Omega)}}
$$
and then 
$$
\Vert u^t\Vert_{H^1(\Omega)}\leq (c(\Omega))^{-1} C(\Omega)[1+\frac{C^2(\Omega)C(\kappa) c(\Omega)}{V_0c(\Omega)-C^2(\Omega)C(\kappa)}] \Vert u^i \Vert_{H^1(\Omega)}.  
$$

\bigskip

Let us now show how to remove the condition on the frequency $\kappa$ in (\ref{condition-kappa}). For this purpose, we go back to the original Lippmann-Schwinger equation

\begin{equation}\label{Lipp-Schwinger-equation-revisited}
(h^2+R(\kappa))u^t=h^2u^{i}\;~~~ \mbox{ in } \Omega
\end{equation}
recalling that
\begin{equation}
R(\kappa) u^{t}(x):=\int_{\Omega}\Phi_\kappa(x, y)V_0u^t(y)dy.
\end{equation}
We rewrite it, by multiplying its both sides by $\varrho^2>0$,  as 
$$
(\varrho^2h^2+\overline{R}(\kappa))u^t=\varrho^2h^2u^{i}\;~~~ \mbox{ in } \Omega
$$
where 
\begin{equation}
\bar{R}(\kappa) u^{t}(x):=\int_{\Omega}\Phi_\kappa(x, y)\varrho^2V_0u^t(y)dy.
\end{equation}

This equation is of the form (\ref{Lipp-Schwinger-equation-revisited}) replacing $h$ by $\varrho h$ and $V_0$ by $\varrho^2 V_0$.
With this setting, under the corresponding condition 
\begin{equation}
 C(\kappa) < \frac{\varrho^2 V_0 c(\Omega)}{C^2(\Omega)}
\end{equation}
we get 
$$
\Vert u^t\Vert_{H^1(\Omega)}\leq (c(\Omega))^{-1} C(\Omega)[1+\frac{C^2(\Omega)C(\kappa)c(\Omega)}{\varrho^2 V_0 c(\Omega)-C^2(\Omega)C(\kappa)}] \Vert u^i \Vert_{H^1(\Omega)}.
$$
Finally, for any fixed $\kappa$, we can choose $\varrho$ large enough so that $C(\kappa) < \frac{\varrho V_0 c(\Omega)}{C^2(\Omega)}$ and we end up with
our needed result. Mainly, for any fixed frequency $\kappa$, there exists a positive constant $C(\Omega, \kappa)$ such that 
$$
\Vert u^t\Vert_{H^1(\Omega)}\leq (c(\Omega))^{-1} C(\Omega)[1+\frac{C^2(\Omega)C(\kappa)c(\Omega)}{\varrho^2 V_0 c(\Omega)-C^2(\Omega)C(\kappa)}]\Vert u^i \Vert_{H^1(\Omega)}=O(1), ~~ h<<1.
$$
With this estimate, we have also the ones of the traces $\Vert u^t\Vert_{H^{\frac{1}{2}}(\partial \Omega)}=O(1)$ and $\Vert \partial_{\nu}u^t\Vert_{H^{-\frac{1}{2}}(\partial \Omega)}=O(1)$, as $h<<1$.
By an integration by parts, we have
$$
\int_{\Omega}\vert \nabla u^t(x) \vert^2 +(h^{-2} V_0-\kappa^2)  \vert u^t(x) \vert^2 dx \leq \Vert \partial_{\nu}u^t\Vert_{H^{-\frac{1}{2}}(\partial \Omega)} \Vert u^t\Vert_{H^{\frac{1}{2}}(\partial \Omega)}=O(1), h<<1.
$$
Hence $\Vert u^t\Vert_{L^2(\Omega)} =O(h),\;~~ h<<1$. By interpolation, we deduce that 
$$
\Vert u^t\Vert_{H^{t}(\Omega)}=O(h^{1-t}), ~~ h<<1, ~~ \mbox{ for } t \in [0, 1]
$$
and, by the trace operator estimates, that
$$
\Vert u^t\Vert_{H^t(\partial \Omega)}=O(h^{\frac{1}{2}-t}), ~~ h<<1, ~~ \mbox{ for } t \leq \frac{1}{2}.
$$

 \section{ Proof of Proposition \ref{Prpopostion-Expansion-Far-fields}}\label{prop-Expansion-Far-fields}
 Using a single layer representation of the solution $$u^s(x, d):=\sum^M_{j=1}\int_{\partial D_j} \Phi_{\kappa}(x, s)\sigma_j(s)ds,$$ one can write the Dirichlet boundary condition in a compact form as 
\begin{eqnarray}\label{cmp-1}
 (\bm{L}+\bm{K})\sigma=-U^{\bm{I}}
\end{eqnarray}
where
$\bm{L}:=(\bm{L}_{mj})_{m,j=1}^{M}$ and $\bm{K}:=(\bm{K}_{mj})_{m,j=1}^{M}$, where
\begin{eqnarray}\label{definition-L_K}
\bm{L}_{mj}=\left\{\begin{array}{ccc}
            \mathcal{S}_{mj} & m=j\\
            0 & else
           \end{array}\right.,
&  \ &\bm{K}_{mj}=\left\{\begin{array}{ccc}
            \mathcal{S}_{mj} & m\neq j\\
            0 & else
           \end{array}\right., 
\end{eqnarray}
\vspace{-.3cm}
\begin{eqnarray}
U^{\bm{I}}=U^{\bm{I}}(s_1,\dots,s_M):=\left(U^i(s_1),\dots,U^i(s_M)\right)^T,\\
 \mbox{ and  }\sigma=\sigma(s_1,\dots,s_M):=\left(\sigma_1(s_1),\dots,\sigma_M(s_M)\right)^T.
 \end{eqnarray}
Here, for the indices  $m$ and $j$ fixed, $\mathcal{S}_{mj}$ is the integral operator acting as

\begin{eqnarray}\label{defofSmjed}
 \mathcal{S}_{mj}(\sigma_j)(t):=\int_{\partial D_j}\hspace{-.07cm}\Phi_{\kappa}(t,s)\sigma_j(s)ds,\quad t\in\partial D_m. 
\end{eqnarray}

\subsection{A priori estimate of $\sigma$}

We can rewrite \eqref{cmp-1} as
\begin{eqnarray}\label{cmp-2}
 (\bm{L}+\acute{\bm{K}})\sigma=-U^{\bm{I}}-\bm{\Phi}^{c}\sigma
\end{eqnarray}

where 

\begin{eqnarray}\label{definition-L_K_prime}
\acute{\bm{K}}_{mj}=\left\{\begin{array}{ccc}
            \mathcal{S}^\prime_{mj} & m=j\\
            0 & else
           \end{array}\right.,
&  \ &\bm{\Phi}^{c}_{mj}=\left\{\begin{array}{ccc}
            \mathcal{S}^c_{mj} & m\neq j\\
            0 & else
           \end{array}\right., 
\end{eqnarray}
with
\begin{eqnarray}\label{defofSmjed_pr}
 \mathcal{S}_{mj}^\prime(\sigma_j)(t)&:=&\int_{\partial D_j}\hspace{-.07cm}[\Phi_{\kappa}(t,s)-\Phi_{\kappa}(z_m,z_j)]\sigma_j(s)ds,\quad t\in\partial D_m\\
\label{defofSmjed_const}
 \mathcal{S}_{mj}^c(\sigma_j)(t)&:=&\int_{\partial D_j}\hspace{-.07cm}\Phi_{\kappa}(z_m,z_j)\sigma_j(s)ds,\quad t\in\partial D_m\nonumber\\
                             &=&\Phi_{\kappa}(z_m,z_j)\,Q_j,  t\in\partial D_m.
\end{eqnarray}

Here $Q_j$ are called total charge on each surface $\partial D_j$, associated to the surface charge distributions $\sigma_j$ and are defined as
\begin{eqnarray}\label{defofQm}
 Q_j:=\int_{ \partial D_j} \sigma_j(s) ds\mbox{ and } Q:=\left({Q_1},{Q_2},\dots,{Q_M}\right)^T.
\end{eqnarray} 
 
Further, \eqref{cmp-2} can be written as,

\begin{eqnarray}\label{cmp-3pr}
\bm{L}_0\sigma=-U^{\bm{I}}-\bm{\Phi}^{c}Q-\acute{\bm{K}}\sigma-(\bm{L}-\bm{L}_0)\sigma 
\end{eqnarray} 

Here, $\bm{L}_0$ is defined same as the matrix operator $\bm{L}$ but for the zero frequency by denoting the corresponding single layer operators by $\mathcal{S}_{mj}^0$.  
By making use of invertibility of $\bm{L}_0$, we can rewrite \eqref{cmp-3pr} as below;
 
 \begin{eqnarray}\label{cmp-4pr}
 \sigma=-\bm{L}_0^{-1}U^{\bm{I}}-\bm{L}_0^{-1}\bm{\Phi}^{c}Q-\bm{L}_0^{-1}\acute{\bm{K}}\sigma+\bm{L}_0^{-1}(\bm{L}-\bm{L}_0)\sigma,
\end{eqnarray}

Observe that  $\|\mathcal{S}_{mm}- \mathcal{S}_{mm}^0\|_{\mathcal{L}\left(L^2(\partial D_m), H^1(\partial D_m)\right)}$ behaves as $O(a^2)$, and hence the norm $\|\bm{L}-\bm{L}_0\|_{\mathcal{L}\left(\prod\limits_{m=1}^{M}L^{2}(\partial D_m),\prod\limits_{m=1}^{M}H^{1}(\partial D_m)\right)}$
 also behaves as $O(a^2)$.

 Next, we show that $\|(S_{mm}^{0})^{-1}S_{mj}^{\prime}\sigma_j\|_{L^2(\partial D_m)}$, for $m\neq\,j$ 
behaves as $O(\frac{a^2}{d_{mj}^2})=C^\prime\frac{a^2}{d_{mj}^2}$ (a uniform constant). Indeed, using the analyticity ( or the Taylor series at the  order $l+1$ if $t\leq\frac{l}{l+1}$, 
see Remark \ref{comments-for-invertibility condition-taylor-analyticity}) of $\Phi_\kappa$ in $\bar{D}_m\times\bar{D}_j$ for $m\neq\;j$ around $z_m$ and then the Taylor series of first order for the first term around $z_l$, we can write
\begin{eqnarray}\label{footnotephik-taylor-analytic-imp}
\Phi_{\kappa}(t,s)&=&\Phi_{\kappa}(z_m,s)+ (t-z_m)\cdot\nabla_t\Phi_{\kappa}(z_m,s) + \sum_{\vert\eta\vert=1}^\infty \frac{D_x^\eta\Phi_{\kappa}(z_m,s)}{\eta!} (t-z_m)^\eta\nonumber\\
&=&\Phi_{\kappa}(z_m,z_j)+(s-z_j)\cdot\int_0^1\nabla_s\Phi_{\kappa}(z_m,s-\alpha(s-z_j))d\alpha+ (t-z_m)\cdot\nabla_t\Phi_{\kappa}(z_m,s)\nonumber\\
&+&  \sum_{\vert\eta\vert=2}^\infty \frac{D_x^\eta\Phi_{\kappa}(z_m,s)}{\eta!} (t-z_m)^\eta.\nonumber\\
\end{eqnarray}
Observe that, the above expansion of $\Phi_{\kappa}$ is due to expansion of infinite order around $z_m$. However, if we just want to use the $(l+1)^{th}$ order expansion of $\Phi_{\kappa}$, then the last two terms of the equation \eqref{footnotephik-taylor-analytic-imp} 
will be replace by $\sum_{\vert\eta\vert=1}^{l+1} \frac{D_x^\eta\Phi_{\kappa}(z_m,s)}{\eta!} (t-z_m)^\eta+\sum_{\vert\eta\vert=l+2} (t-z_m)^\eta \frac{|\eta|}{\eta!}\int_{0}^1 D_t^\eta\Phi_{\kappa}(z_m+\alpha(t-z_m),s) \;d\alpha$.

Now write $S_{mm}^0\psi=(t-z_m)^n$. Then using scaling, we can have $\|\psi\|_{L^2(\partial D_m)}=\epsilon\|\hat{\psi}\|_{L^2(\partial B_m)}$ which gives us $\|(S_{mm}^0)^{-1}((t-z_m)^n)\|_{L^2(\partial D_m)}=O(\epsilon^n)$ and it sufficient to prove the claim $\|(S_{mm}^{0})^{-1}S_{mj}^{\prime}\sigma_j\|_{L^2(\partial D_m)}=O(\frac{a^2}{d_{mj}^2})$. Indeed, 
\begin{eqnarray}\label{Ind-smjinvsmjprsig}
 \|(S_{mm}^{0})^{-1}S_{mj}^{\prime}\sigma_j\|_{L^2(\partial D_m)}
    &\equiv& \|(S_{mm}^{0})^{-1}(\int_{\partial D_j}\hspace{-.07cm}[\Phi_{\kappa}(\cdot,s)-\Phi_{\kappa}(z_m,z_j)]\sigma_j(s)ds)\|_{L^2(\partial D_m)}\nonumber\\
     &\leq& \|(S_{mm}^{0})^{-1}(\int_{\partial D_j}\hspace{-.07cm}(s-z_j)\cdot[\int_0^1\nabla_s\Phi_{\kappa}(z_m,s-\alpha(s-z_j))d\alpha]\,\sigma_j(s)ds)\|_{L^2(\partial D_m)}\nonumber\\
     &&+ \|(S_{mm}^{0})^{-1}(\int_{\partial D_j}\hspace{-.07cm}[(\cdot-z_m)\cdot\nabla_x\Phi_{\kappa}(z_m,s)]\sigma_j(s)ds)\|_{L^2(\partial D_m)}\nonumber\\
      &&+ \|(S_{mm}^{0})^{-1}(\int_{\partial D_j}\hspace{-.07cm}[\sum_{\vert\eta\vert=2}^\infty \frac{D_x^\eta\Phi_{\kappa}(z_m,s)}{\eta!} (\cdot-z_m)^\eta]\sigma_j(s)ds)\|_{L^2(\partial D_m)}\nonumber\\
    &\leq& \|(S_{mm}^{0})^{-1} (1) \|_{L^2(\partial D_m)}\left|\int_{\partial D_j}\hspace{-.07cm}(s-z_j)\cdot[\int_0^1\nabla_s\Phi_{\kappa}(z_m,s-\alpha(s-z_j))d\alpha]\,\sigma_j(s)ds\right|\nonumber\\
     &&+ \|(S_{mm}^{0})^{-1}(\cdot-z_m)\|_{L^2(\partial D_m)}\left|\int_{\partial D_j}\hspace{-.07cm}\nabla_x\Phi_{\kappa}(z_m,s)\sigma_j(s)ds\right|\nonumber\\
      &&+ \sum_{\vert\eta\vert=2}^\infty\frac{1}{\eta!}\|(S_{mm}^{0})^{-1}(\cdot-z_m)^\eta\|_{L^2(\partial D_m)}\left|\int_{\partial D_j}\hspace{-.07cm} {D_x^\eta\Phi_{\kappa}(z_m,s)} \sigma_j(s)ds\right|\nonumber\\
    &=&O\left(\frac{a^2}{d^2_{mj}}\right).
\end{eqnarray}

\bigskip

Let us now estimate  $\|\bm{L}_0^{-1}\acute{\bm{K}}\|$. 
\begin{eqnarray}\label{KDKnrm1}
\left\| \bm{L}_0^{-1}\acute{\bm{K}}\right\|
    &\equiv&
\max\limits_{1\leq m \leq M}\sum_{\substack{j=1\\j\neq\,m}}^{M}\left\| (S_{mm}^{0})^{-1}\mathcal{S}^\prime_{mj}\right\|_{\mathcal{L}\left(L^{2}(\partial D_j),H^{1}(\partial D_m)\right)}\nonumber\\
&\leq&\max\limits_{1\leq m \leq M}\sum_{\substack{j=1\\j\neq\,m}}^{M}C^\prime\frac{a^2}{d_{mj}^2}\nonumber\\
   &\leq&C^\prime\,{M_{\max}}\left[\frac{a^2}{d^2}+\sum_{n=1}^{[a^{-\frac{s}{3}}]}[(2n+1)^3-(2n-1)^3]\left(\frac{a}{n\left(2^{-\frac{1}{3}}a^\frac{s}{3}-\frac{a}{2}\right)}\right)^2\right]\nonumber\\
    &=&C^\prime\,{M_{\max}}\left[\frac{a^2}{d^2}+\left(\frac{a}{\left(2^{-\frac{1}{3}}a^\frac{s}{3}-\frac{a}{2}\right)}\right)^2\sum_{n=1}^{[a^{-\frac{s}{3}}]}[24+\frac{2}{n^2}]\right]\nonumber\\
     &\leq& C^\prime\,{M_{\max}}\left[{a^2}d^{-2}+ 26{a^2}a^{-\frac{s}{3}}\left(2^{-\frac{1}{3}}a^\frac{s}{3}-\frac{a}{2}\right)^{-2}\right]\nonumber\\
     &\leq& C^\prime\,{M_{\max}}\left[{a^2}d^{-2}+ 26{a^2}a^{-s}\right].
\end{eqnarray}

From the above discussions and due to the fact that $\|\bm{L}_0^{-1}\|$ behaves as $O(a^{-1})$, we can observe that 
$\|\bm{L}_0^{-1}(\bm{L}-\bm{L}_0)\|_{\mathcal{L}\left(\prod\limits_{m=1}^{M}L^{2}(\partial D_m),\prod\limits_{m=1}^{M}H^{1}(\partial D_m)\right)}$ and
$\|\bm{L}_0^{-1}\acute{\bm{K}}\|_{\mathcal{L}\left(\prod\limits_{m=1}^{M}L^{2}(\partial D_m),\prod\limits_{m=1}^{M}L^{2}(\partial D_m)\right)}$ 
behaves as $O(a)$ and $O(a^{2-s}+{a^{2-2t}})$ respectively, which enable us to write \eqref{cmp-4pr} as below;

 \begin{eqnarray}\label{cmp-4pr1}
 \sigma=-\bm{L}_0^{-1}U^{\bm{I}}-\bm{L}_0^{-1}\bm{\Phi}^{c}Q+O\left({a}+a^{2-s}+{a^{2-2t}}
\right) \|\sigma\|_{\prod\limits_{m=1}^{M}L^{2}(\partial D_m)}\mbox{ in } L^2,
\end{eqnarray}

Now, for $m=1,\dots,M$, by integrating the $m^{th}$ row of the above equation over $\partial D_m$, we obtain;

 \begin{eqnarray}\label{cmp-4pr2}
Q_m=-\int_{\partial D_m}(\bm{L}_0^{-1}U^{\bm{I}}+\bm{L}_0^{-1}\bm{\Phi}^{c}Q)_m+O\left( {a^2}+a^{3-s}+{a^{3-2t}}\right) \|\sigma\|_{\prod\limits_{m=1}^{M}L^{2}(\partial D_m)}.\quad
\end{eqnarray}

To make the notation simple, introduce the diagonal matrix integral operator $\bm{\int}$ such that $\bm{\int}=Diag(\int_{\partial D_1},\cdots, \int_{\partial D_m})$. Then,  from the above equations we can have;

 \begin{eqnarray}\label{cmp-4pr3}
\left[I+\bm{\int}\bm{L}_0^{-1}\bm{\Phi}^{c}\right]Q=-\bm{\int}\bm{L}_0^{-1}U^{\bm{I}}+O\left({a^2}+a^{3-s}+{a^{3-2t}}\right) \|\sigma\|_{\prod\limits_{m=1}^{M}L^{2}(\partial D_m)}.
\end{eqnarray}

Observe that, $\bm{\int}\bm{L}_0^{-1}\bm{\Phi}^{c}=Diag(\int_{\partial D_1}(\mathcal{S}^0_{11})^{-1},\cdots, \int_{\partial D_m}(\mathcal{S}^0_{mm})^{-1})\bm{\Phi}^{c}$ and hence, $\bm{\int}\bm{L}_0^{-1}\bm{\Phi}^{c}$ is a off-diagonal matrix operator of size $M\times M$ such that
$[\bm{\int}\bm{L}_0^{-1}\bm{\Phi}^{c}]_{pq}=\left\{\begin{array}{ccc}
                                            \int_{\partial D_p}(\mathcal{S}_{pp}^0)^{-1}\Phi_{\kappa}(z_p,z_q), &\mbox{ for } p\neq\;q;\\
                                            0,&\mbox{ for } p=q;
                                           \end{array}\right.$.
Observe that, $\bm{\int}\bm{L}_0^{-1}\bm{\Phi}^{c}Q=Diag(\int_{\partial D_1}(\mathcal{S}_{11})^{-1},\cdots, \int_{\partial D_m}(\mathcal{S}_{mm})^{-1})\,\bm{\Phi}^{c}Q$ and hence, $\bm{\int}\bm{L}_0^{-1}\bm{\Phi}^{c}Q$ is a vector of length $M$ such that
\begin{eqnarray}\label{smpl-L-C-Ph}
 [\bm{\int}\bm{L}_0^{-1}\bm{\Phi}^{c}Q]_{m}&=& \int_{\partial D_m}(\mathcal{S}_{mm}^0)^{-1}(\sum\limits_{q=1\atop {q\neq\;m}}^M\Phi_{\kappa}(z_m,z_q)Q_q), \nonumber\\
 &= &\sum\limits_{q=1\atop {q\neq\;m}}^M\Phi_{\kappa}(z_m,z_q)\;Q_q\int_{\partial D_m}(\mathcal{S}_{mm}^0)^{-1}(1)(s)\;ds, \nonumber\\
  &= &\mathbf{C}\textbf{B}Q. 
\end{eqnarray}
Here,  $\textbf{B}$ and $\mathbf{C}$ are the off-diagonal and diagonal matrices, respectively. These are defined as $B_n(p,q)= \left\{\begin{array}{ccc}\Phi_{\kappa}(z_p,z_q),& p\neq\;q; \\ 0,&p=q;\end{array}\right.,$ and 
$\mathbf{C}=Diag(C_1,\dots, C_m)$, with ${C_m}:=\int_{\partial D_m}(\mathcal{S}_{mm}^0)^{-1}(1)(s)\;ds$ are called as acoustic capacitances.
From \eqref{smpl-L-C-Ph} and the fact that $[I+\mathbf{C}\mathbf{B}]Q=\mathbf{C} [\mathbf{C}^{-1}+\mathbf{B}]{Q}$, we can rewrite \eqref{cmp-4pr3} as;

 \begin{eqnarray}\label{cmp-4pr5}
 [\mathbf{C}^{-1}+\mathbf{B}]{Q}
 &=&-\underbrace{\mathbf{C}^{-1}\bm{\int}\bm{L}_0^{-1}U^{\bm{I}}}_{{=:Y_1}}-\underbrace{\mathbf{C}^{-1}\bm{\int}[\bm{L}_0^{-1}\acute{\bm{K}}\sigma+\bm{L}_0^{-1}(\bm{L}-\bm{L}_0)\sigma]}_{{=:Y_2}}\nonumber\\
&=&O(1)+O\left({a}+a^{2-s}+{a^{2-2t}}\right) \|\sigma\|_{\prod\limits_{m=1}^{M}L^{2}(\partial D_m)}.
\end{eqnarray}

Now, by applying \cite[Lemma 2.22]{C-S:2014} to the system \eqref{cmp-4pr5}, we can prove the invertibility of the matrix $[\mathbf{C}^{-1}+\mathbf{B}]$ for the conditions 
``$\max\limits_{1\leq{m}\leq{M}} |C_m|<\frac{5\pi}{3\gamma}{d} $ and $\gamma:= {\min\atop{j\neq\,m}}\cos(\kappa|z_m-z_j|)\geq0$`` and derive the estimate
   
\begin{equation}\label{fnlinvert-small-ac-2f}
\begin{split}
 \sum_{m=1}^{M}|{Q}_m|^{2}|C_m|^{-1}
\leq4\left( 1-\frac{3\gamma}{5\pi\,d}\max\limits_{1\leq{m}\leq{M}} |C_m|\right)^{-2}\sum_{m=1}^{M}\left|(Y_1+Y_2)_m\right|^2|C_m|.
\end{split}
\end{equation}
which gives us;

\begin{equation}\label{fnlinvert-small-ac-2f-pr}
\begin{split}
 \sum_{m=1}^{M}|{Q}_m|
\leq2\left( 1-\frac{3\gamma}{5\pi\,d}\max\limits_{1\leq{m}\leq{M}} |C_m|\right)^{-1}M\max\limits_{1\leq m \leq M}|{C}_m|\max\limits_{1\leq m \leq M}\left|(Y_1+Y_2)_m\right|
\end{split}
\end{equation}

Observe that,
\begin{eqnarray}\label{fnlinvert-small-ac-pr}
 &&\|\bm{L}_0^{-1}\bm{\Phi}^{c}Q\|_{\prod\limits_{m=1}^{M}L^{2}(\partial D_m)}\nonumber\\
 &=& Max_m \left\|(S_{mm}^0)^{-1}\left(\sum_{j=1,j\neq m}^{M}\Phi_{\kappa}(z_m,z_j)\int_{\partial D_j} \sigma_j\right)\right\|_{L^{2}(\partial D_m)}\nonumber\\
 &\leq& Max_m\|(S_{mm}^0)^{-1}\|_{\mathcal{L}(H^1(\partial D_m), L^2{\partial D_m})}\sum_{j=1,j\neq m}^{M} \|\Phi_{\kappa}(z_m,z_j)\int_{\partial D_j} \sigma_j\|_{H^{1}(\partial D_m)}\nonumber\\
  &\leq& Max_m\epsilon^{-1}\|(S_{mm}^0)^{-1}\|_{\mathcal{L}(H^1(\partial B_m), L^2{\partial B_m})}\sum_{j=1,j\neq m}^{M}|\Phi_{\kappa}(z_m,z_j)|\,|Q_j| \left\|1\right\|_{H^{1}(\partial D_m)}\nonumber\\
      &\leq& Max_m\|(S_{mm}^0)^{-1}\|_{\mathcal{L}(H^1(\partial B_m), L^2{\partial B_m})}|\partial B_m|^\frac{1}{2}\sum_{j=1,j\neq m}^{M}|\Phi_{\kappa}(z_m,z_j)|\,|Q_j| \nonumber\\
  &\leq\atop\eqref{fnlinvert-small-ac-2f-pr}&( Max_m\|(S_{mm}^0)^{-1}\|_{\mathcal{L}(H^1(\partial B_m), L^2{\partial B_m})}|\partial B_m|^\frac{1}{2})2\left(1-\frac{3\gamma}{5\pi\,d}\max\limits_{1\leq m \leq M}|{C}_m|\right)^{-1}\frac{M}{d}\max\limits_{1\leq m \leq M}|{C}_m|\max\limits_{1\leq m \leq M}\left|(Y_1+Y_2)_m\right| \nonumber\\
&\leq&( Max_m\|(S_{mm}^0)^{-1}\|_{\mathcal{L}(H^1(\partial B_m), L^2{\partial B_m})}|\partial B_m|^\frac{1}{2})2\left(1-\frac{3\gamma}{5\pi\,d}\max\limits_{1\leq m \leq M}|{C}_m|\right)^{-1}M\frac{a}{d}\max\limits_{1\leq m \leq M}\left|(Y_1+Y_2)_m\right| \nonumber\\
&=\atop\eqref{cmp-4pr5}&O(M\frac{a}{d})+O\left( {a^{2-s-t}}+a^{3-2s-t}+{a^{3-s-3t}}\right) \|\sigma\|_{\prod\limits_{m=1}^{M}L^{2}(\partial D_m)}.
  \end{eqnarray}

From, \eqref{cmp-4pr1}, we dedue that
 \begin{eqnarray}\label{cmp-4pr6}
 \|\sigma\|_{\prod\limits_{m=1}^{M}L^{2}(\partial D_m)}
 &\leq&\|\bm{L}_0^{-1}U^{\bm{I}}\|_{\prod\limits_{m=1}^{M}L^{2}(\partial D_m)}+\|\bm{L}_0^{-1}\bm{\Phi}^{c}Q\|_{\prod\limits_{m=1}^{M}L^{2}(\partial D_m)}+O\left(a+a^{2-s}+{a^{2-2t}}\right) \|\sigma\|_{\prod\limits_{m=1}^{M}L^{2}(\partial D_m)}\nonumber\\
&\leq\atop\eqref{fnlinvert-small-ac-pr}&O(1)+O(M\frac{a}{d})+O\left({a^{2-s-t}}+a^{3-2s-t}+{a^{3-s-3t}}\right) \|\sigma\|_{\prod\limits_{m=1}^{M}L^{2}(\partial D_m)}\nonumber\\
&&\qquad\qquad+O\left(a+a^{2-s}+{a^{2-2t}}\right) \|\sigma\|_{\prod\limits_{m=1}^{M}L^{2}(\partial D_m)}.
 \end{eqnarray}
 
 Hence $\|\sigma_m\|_{L^2(\partial D_m)}$ behaves as $O\left(1+M\frac{a}{d}\right)$, if we have ${a}$, ${a^{2-s-t}}$, $a^{2-s}$, $a^{3-2s-t}$, ${a^{2-2t}}$, ${a^{3-s-3t}}$ are small enough. Hence, we  should have 
 ${2-s-t}$, $2-s$, $3-t-2s$, ${2-2t}$ and ${3-s-3t}$ are {positive}. Due to the fact that $t\geq\frac{s}{3}$, the above inequalities are equivalent to the following conditions on $t$ and $s$.
 \begin{eqnarray}\label{condtions-ts-neumann} 
   \frac{s}{3}\leq{t}\leq{1}&\mbox{ and } 0\leq{s}\leq\min\left\{2,{2-t}, {3-3t},\frac{1}{2}(3-t)\right\}
 \end{eqnarray}

 For simplicity, let $C_\sigma$ be the uniform constant such that
 $\|\sigma_m\|_{L^2(\partial D_m)}\leq\|\sigma\|_{\prod\limits_{m=1}^{M}L^2(\partial D_m)}\leq C_{\sigma}\left(1+M\frac{a}{d}\right).$

 \subsection{Derivation of the asymptotic Expansion}
 Using the $L^2$ estimate of $\sigma$, we have following estimates for the total charge $Q_m$ and the far-field $U^\infty$;
\begin{eqnarray}\label{estofQm}
 |Q_m|&\leq&{C_\sigma\,a\left(1+M\frac{a}{d}\right)}, \mbox{ for } m=1,\dots,M,
\end{eqnarray}
\begin{equation}\label{x oustdie1 D_m}
U^\infty(\hat{x})=\sum_{m=1}^{M}[e^{-i\kappa\hat{x}\cdot z_{m}}Q_m+O({\kappa\,a^2\left(1+M\frac{a}{d}\right)})],
\end{equation}
with $Q_m$ given by (\ref{defofQm}), if $\kappa\,a<1$ where $O\left(\kappa\,a^2\left(1+M\frac{a}{d}\right)\right)\,\leq\,{C_\sigma\,\kappa\left(1+M\frac{a}{d}\right)a^2}$.

Again by making use of $L^2$ estimate of $\sigma$, we can rewrite \eqref{cmp-4pr5} as below, for $m=1,\dots,M$,
  \begin{align}\label{cmp-4pr1p1}
\frac{Q_m}{C_m}+\sum_{j\neq\,m}C_j\Phi_\kappa(z_m,z_j)\frac{Q_j}{C_j}&=-\frac{1}{C_m}\int_{\partial D_m}(\bm{L}_0^{-1}U^{\bm{I}})_m+O\left({a}+{a^{2-s-t}}+ a^{2-s}+a^{3-2s-t}+{a^{2-2t}}+{a^{3-s-3t}}\right)\nonumber\\
&=-\frac{1}{C_m}\int_{\partial D_m}((S_{mm}^0)^{-1}U^{\bm{I}})(s_m)ds_m+O\left({a}+{a^{2-s-t}}+ a^{2-s}+a^{3-2s-t}+{a^{2-2t}}+{a^{3-s-3t}}\right)\nonumber\\
&=-\frac{1}{C_m}\int_{\partial D_m}(S_{mm}^0)^{-1}(U^i(z_m))(s_m)ds_m-\frac{1}{C_m}\int_{\partial D_m}(S_{mm}^0)^{-1}[U^i(s_m)-U^i(z_m)]ds_m\nonumber\\
&\qquad+O\left({a}+{a^{2-s-t}}+ a^{2-s}+a^{3-2s-t}+{a^{2-2t}}+{a^{3-s-3t}}\right)\nonumber\\
&=-U^i(z_m)+O\left({a}+{a^{2-s-t}}+ a^{2-s}+a^{3-2s-t}+{a^{2-2t}}+{a^{3-s-3t}}\right).
\end{align}

For $m=1,\dots,M$, let $\bar{Q}_m$ be the potentials such that, 
  \begin{eqnarray}\label{cmp-4pr1p2}
\frac{\bar{Q}_m}{C_m}+\sum_{j\neq\,m}C_j\Phi_\kappa(z_m,z_j)\frac{\bar{Q}_j}{C_j}
&=&-U^i(z_m).
\end{eqnarray}

Now by taking the difference between \eqref{cmp-4pr1p1} and \eqref{cmp-4pr1p2}, for $m=1\dots,M$, provide us;
  \begin{align}\label{cmp-4pr1p3}
\frac{Q_m-\bar{Q}_m}{C_m}&+\sum_{j\neq\,m}C_j\Phi_\kappa(z_m,z_j)\frac{Q_j-\bar{Q}_j}{C_j}
\nonumber\\
&=O\left({a}+{a^{2-s-t}}+ a^{2-s}+a^{3-2s-t}+{a^{2-2t}}+{a^{3-s-3t}}\right).
\end{align}

Comparing this system of equations with \eqref{cmp-4pr5}, we obtain;

\begin{eqnarray}\label{fnlinvert-small-ac-2f-pr1}
 \sum_{m=1}^{M}|{Q}_m-\bar{Q}_m|
&\leq&2\left( 1-\frac{3\gamma}{5\pi\,d}\max\limits_{1\leq{m}\leq{M}} |C_m|\right)^{-1}M\max\limits_{1\leq m \leq M}|{C}_m|\left|O\left({a}+{a^{2-s-t}}+ a^{2-s}+a^{3-2s-t}+{a^{2-2t}}+{a^{3-s-3t}}\right)\right|\nonumber\\
&=&O\left(a^{2-s}+{a^{3-2s-t}}+ a^{3-2s}+a^{4-3s-t}+{a^{3-s-2t}}+{a^{4-2s-3t}}\right).
\end{eqnarray}

We can rewrite the far-field \eqref{x oustdie1 D_m} as,

\begin{eqnarray}\label{x oustdie1 D_m1}
U^\infty(\hat{x})&=&\sum_{m=1}^{M}[e^{-i\kappa\hat{x}\cdot z_{m}}Q_m+O({\kappa\,a^2\left(1+M\frac{a}{d}\right)})]\nonumber\\
&=&\sum_{m=1}^{M}e^{-i\kappa\hat{x}\cdot z_{m}}Q_m+O({a^{2-s}+a^{3-2s-t}})\nonumber\\
&=&\sum_{m=1}^{M}e^{-i\kappa\hat{x}\cdot z_{m}}\bar{Q}_m+\sum_{m=1}^{M}e^{-i\kappa\hat{x}\cdot z_{m}}(Q_m-\bar{Q}_m)+O({a^{2-s}+a^{3-2s-t}})\nonumber\\
&{=\atop\eqref{cmp-4pr1p3}}&\sum_{m=1}^{M}e^{-i\kappa\hat{x}\cdot z_{m}}\bar{Q}_m+O\left({a^{2-s}+a^{3-2s-t}+a^{3-2s}+ a^{4-3s-t}}+{a^{3-s-2t}}+{a^{4-2s-3t}}\right)\nonumber\\
&{=}&\sum_{m=1}^{M}e^{-i\kappa\hat{x}\cdot z_{m}}\bar{Q}_m+O\left({a^{2-s}+a^{3-2s-t}+ a^{4-3s-t}}+{a^{3-s-2t}}+{a^{4-2s-3t}}\right).\qquad
\end{eqnarray}
\bigskip

\begin{remark}\label{comments-for-invertibility condition-taylor-analyticity}
 
  Instead of using the analyticity of $\Phi_{\kappa}$ around $z_m$ in \eqref{KDKnrm1}, we use the Taylor series expansion till the order ${(l+1)}$, $l=0,1,2,\cdots$, 
  then following similar computations as in \eqref{Ind-smjinvsmjprsig}, we can prove that   $\|(S_{mm}^{0})^{-1}S_{mj}^{\prime}\sigma_j\|_{L^2(\partial D_m)}$, for $m\neq\,j$,
behaves as $O(\frac{a^2}{d_{mj}^2}+\frac{a^{2+l}}{d_{mj}^{3+l}})$ and then deduce that
$\|\bm{L}_0^{-1}\acute{\bm{K}}\|_{\mathcal{L}\left(\prod\limits_{m=1}^{M}L^{2}(\partial D_m),\prod\limits_{m=1}^{M}L^{2}(\partial D_m)\right)}$ behaves as
$O(a^{(2+l)-(4+l)\frac{s}{3}}+{a^{(2+l)-(3+l)t}}+a^{3-{s}}+{a^{(2-2t}})$, which eventually leads to the estimate $\|\sigma_m\|_{L^2(\partial D_m)}=O(1+M\frac{a}{d})$ under the following assumptions

 \begin{eqnarray}\label{condtions-ts-neumann-nopower}
   \frac{s}{3}\leq{t}\leq{\frac{2+l}{3+l}}&\mbox{ and } 0\leq{s}\leq\min\left\{3\frac{(2+l)}{(4+l)},{2-t}, {(3+l)-(4+l)t},\frac{3}{7+l}(3+l-t), 2, {3-3t},\frac{1}{2}(3-t) \right\}.\nonumber\\
 \end{eqnarray}
 
  \end{remark}

 \bibliographystyle{abbrv}

\end{document}